\documentclass[letterpaper,12pt]{article}
\usepackage[utf8]{inputenc}
\usepackage{natbib,amsmath, color, caption, subfig, mathtools,amsfonts,amsthm, titlesec, setspace, authblk,amsfonts,graphicx}
\usepackage[letterpaper, margin=1in]{geometry}
\usepackage{epstopdf}
\newtheorem{proposition}{Proposition}
\pdfoutput=1

\titleformat*{\section}{\large\bfseries}
\title{ \bf {Regression Adjustment for Noncrossing Bayesian Quantile Regression}}
\date{}
\author{T. Rodrigues\footnote{CAPES Foundation, Ministry of Education of Brazil, Bras\'{i}lia - DF 70040-020, Brazil \\ 
\indent \;\, Communicating Author: {\tt t.rodrigues@unsw.edu.au}} \, and Y. Fan \\
School of Mathematics and Statistics, \\
University of New South Wales\\
Sydney 2052, Australia}

\begin{document}

\maketitle

\begin{abstract}
A two-stage approach is proposed to overcome the problem in quantile regression, where separately fitted curves for several quantiles may cross. 
The standard Bayesian quantile regression model is applied in the first stage, followed by a Gaussian process regression adjustment, which monotonizes the quantile function whilst borrowing strength from nearby quantiles. The two stage approach is computationally efficient, and more general than existing techniques. The method is shown to be competitive with alternative approaches via its performance in simulated examples. 


\end{abstract}

\newpage
\section{Introduction}

In many applications, interest lies in describing the effect of a set of covariates at the tail of the response distribution, which can be considerably different from their impact at the mean \citep{Koenker2005}. The need to have a wider picture of the conditional distribution of the response variable is the main reason for the origin and popularity of quantile regression methods. These models have been applied to many areas, including the environmental sciences; medicine; engineering and economics. Often in the context of risk assessment, where the tail of the distribution plays an important role. In many cases, quantile estimates at several different quantile levels are needed, so when estimation for each level is carried out separately, the monotonicity of the conditional quantiles can be violated, giving rise to the phenomenon of crossing quantile regression curves. This leads to difficulties for inference, since by definition, the conditional quantile function should be monotonically increasing. Figure~\ref{igg1} shows the estimated conditional quantile function using the Immunoglobulin-G (IgG) data set of \cite{Isaac1983}, with age as the regressor. The figure shows the quantile function at age 6. The dashed line shows the estimates obtained from separately fitting the quantile regression at several levels, and it is evident that the curve is not monotone, particularly at the extremal levels. The solid line demonstrates the correction obtained using regression adjustment for monotonicity proposed later in this article.

\begin{figure}[!ht]
 \centering
 \includegraphics[width=7.5cm,height=7.5cm,angle=-90]{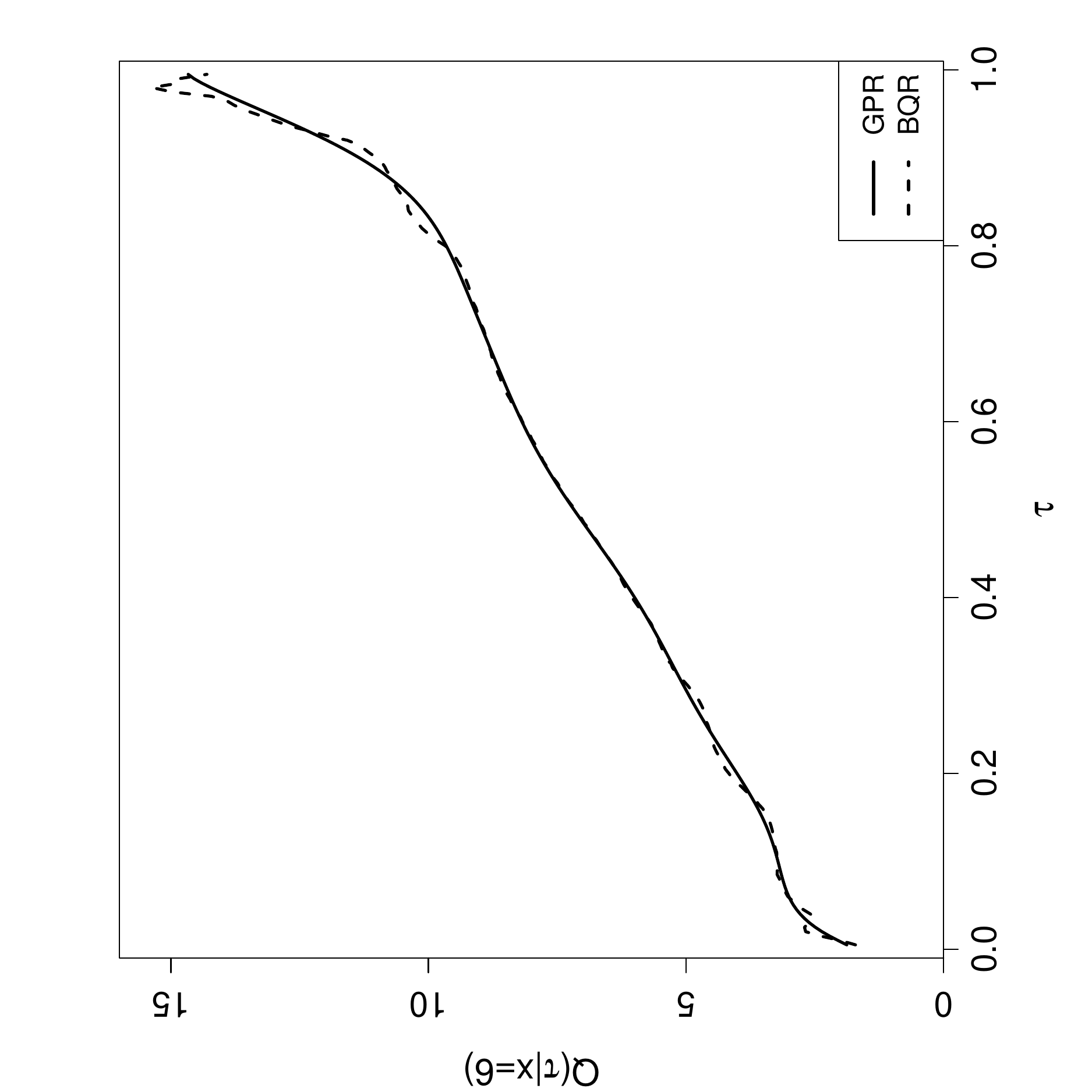} 
 \caption{Estimated conditional quantile function of serum concentration of IgG for children at age 6, using standard Bayesian quantile regression (dashed line, BQR) and regression adjusted estimates (solid line, GPR).%
 }
 \label{igg1}
\end{figure}

Quantile regression models were introduced in a seminal paper by \cite{Koenker1978}. By writing the sample quantiles as an optimisation problem, and generalising to the linear regression model, they proposed minimising the loss function 
\begin{equation*}
 \sum_{i=1}^{n}{\rho_{\tau}(y_i-\mathbf{X_i}^T\mathbf{\beta})},
\end{equation*}
where $\rho_{\tau}$ is the check function $\rho_{\tau}(u)=u(\tau-I(u<0))$. The solution to this minimisation yields the $\tau$-th quantile regression estimate. Most of frequentist literature in quantile regression is based on this estimator, which does not make probabilistic assumptions for the response variable. Consequently, inference for the parameters of interest relies on its asymptotic properties (see \cite{Koenker2005} and references therein). In the Bayesian framework,  motivated by its equivalence to the minimization problem, \cite{Yu2001} proposed the use of the asymmetric Laplace distribution (ALD) as an approximation to the likelihood. The approach is appealing since even when the true distribution of the data was not ALD, empirical results were satisfactory. \cite{Sriram2013} showed that under some mild conditions posterior consistency can be established for the linear quantile regression estimates based on the ALD, providing some theoretical support for the model.

A drawback of the above approaches is that, as quantiles are fitted separately, the conditional quantile curves are not smooth and the fitted regression lines may cross, which violates the basic probabilistic rule and causes problems for inference in practice. To overcome crossing, \cite{He1997} restricted the space of possible solutions of the response distribution to location-scale changes of a base distribution to get noncrossing curves. \cite{Yu1998} and \cite{Hall1999}, among others, proposed to estimate the conditional distribution function nonparametrically. \cite{Dette2008} and \cite{Victor2009} proposed to monotonize an estimated conditional distribution function and then invert it to obtain the quantiles. A clever approach was presented by \cite{Bondell2010} to simultaneously estimate several quantile levels by solving the constrained minimisation problem. They showed that in the linear case, the number of constraints can be greatly reduced, leading to an efficient estimation algorithm. The method is extendable to linear splines.

Most solutions to noncrossing quantiles in the Bayesian literature proceed by simultaneously fitting several quantiles. In the case of linear quantile regression, including spatially correlated data, \cite{Reich2011} were able to simplify the noncrossing constraints by writing the quantile process using Bernstein basis polynomials. Nevertheless, the likelihood does not have a closed form and for moderately sized data sets the proposed method is infeasible, so an adjustment to the classical estimates is also suggested. For linear regression with a single predictor, or a single index model, \cite{Surya2012} suggested a reparametrization of the quantile function that induces monotonicity. Recently, an ingenious solution for linear quantile regression was presented by \cite{Reich2013}, where the quantile function is modelled piece-wise using the linear heteroscedastic model. While many of these approaches work well in certain instances, they lack flexibility for more complex situations. Alternative Bayesian nonparametric methods have also been proposed (see \cite{Scaccia2003} and \cite{Taddy2010}, among others). They specify flexible error distributions using Bayesian nonparametric techniques, but model simplicity is generally compromised. As an alternative, \cite{Dunson2005}, and similarly \cite{Lancaster2010}, proposed a substitution likelihood for approximate simultaneous quantile inference.

In this article, we adopt a fully Bayesian approach, thus inference does not rely on asymptotic results, which can be particularly delicate in the quantile crossing context where sample sizes may not be big enough. We propose a two-stage approach for the simultaneous estimation of quantile regression at multiple levels. The first stage uses standard Bayesian quantile regression with ALD, fitted separately at different quantile levels. This takes advantage of the flexibility of modelling afforded by fitting quantiles for a specific level. These initial estimates are then adjusted by borrowing strength across nearby quantiles using Gaussian process regression in the second stage. The noncrossing constraints are controlled through a single parameter and no MCMC is needed in the second stage, resulting in a simple and efficient algorithm. 
The advantages of our approach are that our proposed method produces smoother estimates of the conditional quantile functions (see Figure~\ref{igg1}), and can handle complex modelling situations through the use of ALD in the first stage inference.

The rest of the article is organized as follows. Section 2 introduces the two-stage model for adjusted Bayesian quantile regression, along with its properties and estimation procedure. Simulations are performed in Section 3 to extensively compare the performance of the proposed method with the best current solutions. In Section 4, we analyse two real examples, the famous data set of serum concentration of immunoglobulin-G for young children and the global mean sea level time series. The final section offers some concluding discussions.

\section{Bayesian Quantile Regression Adjustment}
This section addresses the two-stage approach for noncrossing Bayesian quantile regression. We first present standard Bayesian quantile regression with ALD in Section 2.1, and then describe the second stage Gaussian process regression adjustment in Section 2.2.

\subsection{First stage: Standard Bayesian quantile regression}
Let $\mathbf{y}=(y_1, ..., y_n)$ be a vector of observed random variables 
from an unknown true distribution $f$. Associated with each observation, we have a $k$-dimensional vector of covariates $\mathbf{X}_i \in \mathbb{R}^k$. For quantile level $0<\tau<1$, the $\tau$-th quantile regression model is given by
\begin{equation}
 y_i | \mathbf{X}_i = h_{\tau}(\mathbf{X}_i) + \epsilon_i, \quad i=1, \cdots, n
 \label{model}
\end{equation}
where $\epsilon_i$ are independent from a distribution whose $\tau$-th quantile equals to $0$, i.e. $P(\epsilon_i \leq 0)=\tau$, and $h_{\tau}$ is an arbitrary function of the covariates.

The standard Bayesian estimation procedure is employed in the first stage. Here we consider the use of the auxiliary likelihood given  by the ALD($p$) distribution as an approximation to the true likelihood where
\begin{equation}
L_a(\mathbf{y}|\boldsymbol{\mu},\sigma,p)=\frac{p^n(1-p)^n}{\sigma^n}\exp{\left\{- \sum_i{ \rho_{p}\left(\frac{y_i- \mu_i}{\sigma}\right)} \right\}},
\label{ald}
\end{equation}
where the mean $\mu_i\equiv h_{\tau}(\mathbf{X}_i)$ and scale parameter $\sigma$ are estimated. The value $p$ corresponds to the value of the quantile. If we are interested in the $\tau$th quantile, the ALD($p=\tau$) is fitted to provide a good approximation for the true quantile function $Q(\tau|x)$ locally at $\tau$. For several quantile levels, $\tau=p_1,\ldots p_P$, the procedure is carried out for each value of $\tau$ using their respective ALD distributions. It is in this last procedure that the phenomenon of crossing quantile curves are often observed, as the curves are fitted independently, using different likelihood functions.

Posterior estimation of unknown parameters for a given $\tau$ is typically obtained via MCMC methods. For the $\tau$-th conditional quantile under auxiliary model ALD($p=\tau$), let $Q^{(t)}(\tau|\mathbf{X},p=\tau), t=1,...,T$, denotes the $t$-th posterior quantile estimate given by the MCMC samples.
Following Bayesian model averaging (BMA) approach, the $\tau$-th quantile point estimate is given by    
\begin{equation}
\widehat{Q}_s(\tau|x)= \frac{1}{T}\sum_{t=1}^T{Q^{(t)}(\tau|\mathbf{X},p=\tau)},
\label{post_mean}
\end{equation}
where the index $s$ denotes the standard estimate. The standard formulation of the first stage encompasses any model written as \eqref{model} and \eqref{ald}, our simulations examples later will consider both linear and non-linear regression models.

\subsection{Second stage: Gaussian process regression adjustment}
Standard quantile estimate \eqref{post_mean} is computed based only on the MCMC samples from the auxiliary model ALD($p=\tau$). However, as adjacent quantiles are correlated, it is expected that other auxiliary models ALD($p=\tau'$) will contain useful information for $\tau$ if $\tau'$ is nearby. Since in the first stage we have fitted $P$ auxiliary models ALD($p$), $p=p_1, ..., p_P$, we can calculate the $induced$ $\tau$-th quantile posterior sample for any given auxiliary model from
 \begin{equation}
  Q^{(t)}(\tau|\mathbf{X},p) = F^{-1}(\tau; \mu^{(t)},\sigma^{(t)},p) =
  \begin{dcases}
   \mu^{(t)} + \frac{\sigma^{(t)}}{1-p}\log{\left(\frac{\tau}{p}\right)},    & \text{if } 0\leq \tau \leq p \\
   \mu^{(t)} - \frac{\sigma^{(t)}}{p}\log{\left(\frac{1-\tau}{1-p}\right)}, & \text{if } p\leq \tau \leq 1
  \end{dcases} \; ,
  \label{induced}
\end{equation}
which is the quantile function of the fitted auxiliary model ALD($\mu^{(t)},\sigma^{(t)},p$) \citep{Yu2005}.
Equation \eqref{induced} provides us with $(P{-}1){\times}T$ additional posterior samples for the quantile at $\tau$ (assuming that one of the $p_1,\ldots p_P$ values equals to $\tau$).
A smoother noncrossing quantile estimate can then be obtained by borrowing strength from these induced $\tau$-th quantile posterior samples using Gaussian process regression on all $P{\times}T$ estimates \citep{GPbook}, 
\begin{equation} 
\label{gp_model}
\begin{aligned}
Q^{(t)}(\tau|\mathbf{X}&,p) = g(p)+ \epsilon \\
\epsilon &\sim \mathcal{N}(0,\Sigma) \\
g(p) &\sim \mathcal{GP}(0,K) ,
\end{aligned}
\end{equation}
where $\Sigma$ and $K$ are covariance matrices with dimension $(P{\times}T,P{\times}T)$. More 
specifically, $\Sigma$ is a diagonal covariance matrix whose diagonal entries are the posterior 
variances of the corresponding $Q^{(t)}(\tau|\mathbf{X},p)$, denoted by 
$\sigma^2(\tau|\mathbf{X},p)$. Independence is assumed between samples of the induced models as 
these were obtained independently, and approximate independence within each model is also attained 
by taking every $m$th MCMC sample. Moreover, as auxiliary models ALD(p), for p close to $\tau$, 
carry more information about the $\tau$-th quantile than more distant ones, we build the Gaussian 
process covariance matrix as a decreasing function of the distance between the models. Hence, using 
the squared exponential kernel, covariance matrix entries between induced quantiles from any two 
auxiliary models are given by
\begin{equation}
k(p,p')=\sigma^2_k\exp{\left\{-\frac{1}{2b^2}(p-p')^2 \right\} },
\label{kernel}
\end{equation}
where $b$ is the bandwidth and $\sigma^2_k$ is a variance hyperparameter of the prior. As we are centering the prior arbitrarily on zero, $\sigma^2_k$ should be large to result in an uninformative prior. However, choosing values too big may lead to computational issues while inverting matrices. Throughout the article, we adopted $\sigma^2_k=100$, results were not sensitive to any moderately large value of $\sigma^2_k$.

The Gaussian process prior over the posterior mean function $g(p)$ allows us to model nonparametrically the embedded correlation structure, using information from other auxiliary models to adjust the standard posterior mean \eqref{post_mean}. The final $\tau$-th quantile estimate will then be the adjusted posterior mean for standard auxiliary model ALD($p=\tau$), denoted by $\widehat{Q}_a(\tau|x)$, which is very simple to obtain as closed form predictive posterior distribution is available \citep{GPbook},
\begin{align*}
Q_*(\tau|\mathbf{X},p_*{=}\tau, \mathbf{Q}^{(t)}) &\sim \mathcal{N} (\mu_*, \sigma_*^2) ,
\end{align*}
where
\begin{subequations}
\begin{align}
\mu_*&= \widehat{Q}_a(\tau|x) = \sum_{p=1}^P \sum_{t=1}^T w_{p} Q^{(t)}(\tau|\mathbf{X},p) , \label{final_quant} \\
\sigma_*^2 &= k(\tau,\tau) - W K(.,\tau) + \sigma^2(\tau|\mathbf{X},p=\tau), \label{final_var} \\
W &= K(.,\tau)^{\top} (K + \Sigma)^{-1} , \label{weights}
\end{align}
\end{subequations}
also $K(.,\tau)$ is a covariance matrix column where $p'=\tau$ and $w_{p}$ is an element of the row vector of weights $W \in \mathbb{R}^{PT}$.

Regardless of the number of quantiles being estimated and their initial values, noncrossing quantile estimates can always be obtained from this Gaussian process regression adjustment, as stated in Proposition \ref{prop1}.

\begin{proposition}
\label{prop1}
Let $\widehat{Q}_a(\tau|x)$ be the adjusted $\tau$-th quantile estimate given in \eqref{final_quant}. Then for any set of quantiles $\tau_1 < ... < \tau_P$, $P \in \mathbb Z_+^*$, there always exist a bandwidth $b$ such that $\widehat{Q}_a(\tau_1|x) \leq ... \leq \widehat{Q}_a(\tau_P|x)$.
\end{proposition}

\begin{proof}
As $b \rightarrow \infty$, weights are all equal and the adjusted $ \tau$-th quantile estimate \eqref{final_quant} can be written as
\begin{align*}
\mu_*= \widehat{Q}_a(\tau|x) =  \frac{1}{P \times T}\sum_{p=1}^P \sum_{t=1}^T Q^{(t)}(\tau|\mathbf{X},p).
\end{align*}
Being a simple average of nondecreasing quantile functions $Q^{(t)}(\tau|\mathbf{X},p)$, the final quantile function is also nondecreasing for $b \rightarrow \infty$. In fact, for the same reason, monotonicity holds when the weight function is constant for the interval $0<p<1$, which is generally achieved for moderately large $b$.
\end{proof}

Although the monotonicity guarantee presented in Proposition \ref{prop1} is important, in general, we do not need to restrict the quantile function to this limiting case to get rid of the crossing. In practice, a small bandwidth $b$ already adds enough smoothness to prevent crossing and is preferable in order to keep the weights concentrated on the target model. In fact, from Equations \eqref{kernel} and \eqref{weights}, it is easy to see that, if the bandwidth $b$ goes to zero, the weights are nonzero only for the target model ($p=\tau$) and $\widehat{Q}_a(\tau|x)=\widehat{Q}_s(\tau|x)$, which is the standard posterior mean. Therefore, the bandwidth, which is the only unknown parameter of Gaussian process model \eqref{gp_model}, will be estimated as the minimum value of $b$ that ensures noncrossing everywhere (for all $\tau$ and $\mathbf{X}$). Consequently, troublesome noncrossing constraints are handled easily here by means of a single smoothing parameter. Moreover, posterior consistency holds if the initial model is consistent, this is stated in Proposition \ref{prop3}.

\begin{proposition}
\label{prop3}
Let $\widehat{Q}_s(\tau|x)$ be the standard $\tau$-th quantile estimate \eqref{post_mean} and $\widehat{Q}_a(\tau|x)$ the final estimate \eqref{final_quant}. Then, if $\widehat{Q}_s(\tau|x)$ is consistent, $\widehat{Q}_a(\tau|x)$ is also consistent.
\end{proposition}

\begin{proof}
If $\widehat{Q}_s(\tau|x)$ is consistent, it does not cross as sample size goes to infinity. Therefore, asymptotically, the estimated bandwidth $b$ goes to zero and the final estimate reduces to the initial one.
\end{proof}

Estimation of regression model \eqref{gp_model} can be further simplified as the weights given in \eqref{weights} do not depend on iteration number $t$. Therefore, instead of using all MCMC posterior samples $Q^{(t)}(\tau|\mathbf{X},x)$, the same final estimates can be obtained by fitting a Gaussian process only to the induced posterior means, avoiding Gaussian process issues with large data sets. The result is presented in Proposition \ref{prop2} (see Appendix for proof).

\begin{proposition}
\label{prop2}
Let $\widehat{Q}_s(\tau|\mathbf{X},p)$ be the posterior mean of induced $\tau$-th quantile from auxiliary model ALD(p),
\begin{equation} 
\widehat{Q}_s(\tau|\mathbf{X},p) = \frac{1}{T}\sum_{t=1}^T{Q^{(t)}(\tau|\mathbf{X},p)}.
\label{qmean}
\end{equation}
Then adjusted posterior mean \eqref{final_quant} and variance \eqref{final_var} can be obtained by fitting a Gaussian process only to the induced posterior means $\widehat{Q}_s(\tau|\mathbf{X},p)$, i.e.
\begin{equation} 
\label{gp_model2}
\begin{aligned}
\widehat{Q}_s(\tau|\mathbf{X}&,p) = g'(p) + \epsilon' \\
\epsilon' &\sim \mathcal{N}(0,\Sigma') \\
g'(p) &\sim \mathcal{GP}(0,K') ,
\end{aligned}
\end{equation}
where $\Sigma'$ is a $(P,P)$ diagonal covariance matrix whose diagonal entries are $\sigma^2(\tau|\mathbf{X},p)/T$, since independence is assumed, and $K'$ is a $(P,P)$ covariance matrix whose
entries are calculated as in \eqref{kernel}. More specially, if we let $g'(p_*{=}\tau) \sim \mathcal{N} (\mu', {\sigma'}^2)$ be the predictive posterior distribution for the mean function of model \eqref{gp_model2}, then the equivalence is as follow
\begin{equation}
\begin{aligned}
\label{final_par}
\mu_* &=\mu' = \sum_{p=1}^P w'_{p} \widehat{Q}_s(\tau|\mathbf{X},p), \\
\sigma_*^2 &= {\sigma'}^2 + \sigma^2(\tau|\mathbf{X},p=\tau),  \\
\end{aligned}
\end{equation}
where ${\sigma'}^2= k'(\tau,\tau) - W' K'(.,\tau)$ and $W'= K'(.,\tau)^{\top} (K' + \Sigma')^{-1}$.
\end{proposition}

Hence, the final quantile estimate is a weighted average of the induced quantile posterior means \eqref{final_par}. Furthermore, the posterior variance \eqref{final_par} suffers no substantial change, as the additional term ${\sigma'}^2$ is the predictive posterior mean variance, being in the order of magnitude of the Monte Carlo variance (${\sigma'}^2\approx\sigma^2/T$), which decreases to zero at the rate 1/T.

We illustrate the estimation process in Figure~\ref{idea_ald}. We consider estimating the conditional deciles for a given data set. After fitting standard Bayesian quantile regression in the first stage, if we look at a given $\mathbf{X}$, we have a set of estimated conditional quantile functions, one for each auxiliary model. Focusing on the $0.4$-th conditional quantile, Figure~\ref{idea_ald_a} shows the $0.4$-th induced quantile posterior means $\widehat{Q}_s(0.4|\mathbf{X},p)$. Borrowing strength is then achieved by fitting a Gaussian process regression with these induced posterior quantiles (Figure~\ref{idea_ald_b}), where the bandwidth is chosen as the minimum value such that noncrossing holds everywhere (for all $\tau$ and $\mathbf{X}$). A general algorithm can be summarized as follows: 

\begin{enumerate}
 \item Fit $P$ separate ALD(p), $p=p_1, ... , p_P$, (Equation \ref{ald}); 
 \item Calculate induced quantile posterior means $\widehat{Q}_s(\tau|\mathbf{X},p)$ for all $\mathbf{X}$  and $\tau=p_1,\ldots, p_P$ (\ref{qmean});
 \item Initialize $b \approx 0$ and while quantile estimates cross, increase $b$ and calculate regression adjusted quantile estimates (\ref{final_par}) for every $\mathbf{X}$ and  $\tau=p_1,\ldots, p_P$.
\end{enumerate}

\begin{figure}[!ht]
    \centering
    \subfloat[\footnotesize{First stage: ALD auxiliary models} \label{idea_ald_a}]{\includegraphics[width=7.5cm,height=7.5cm,angle=-90]{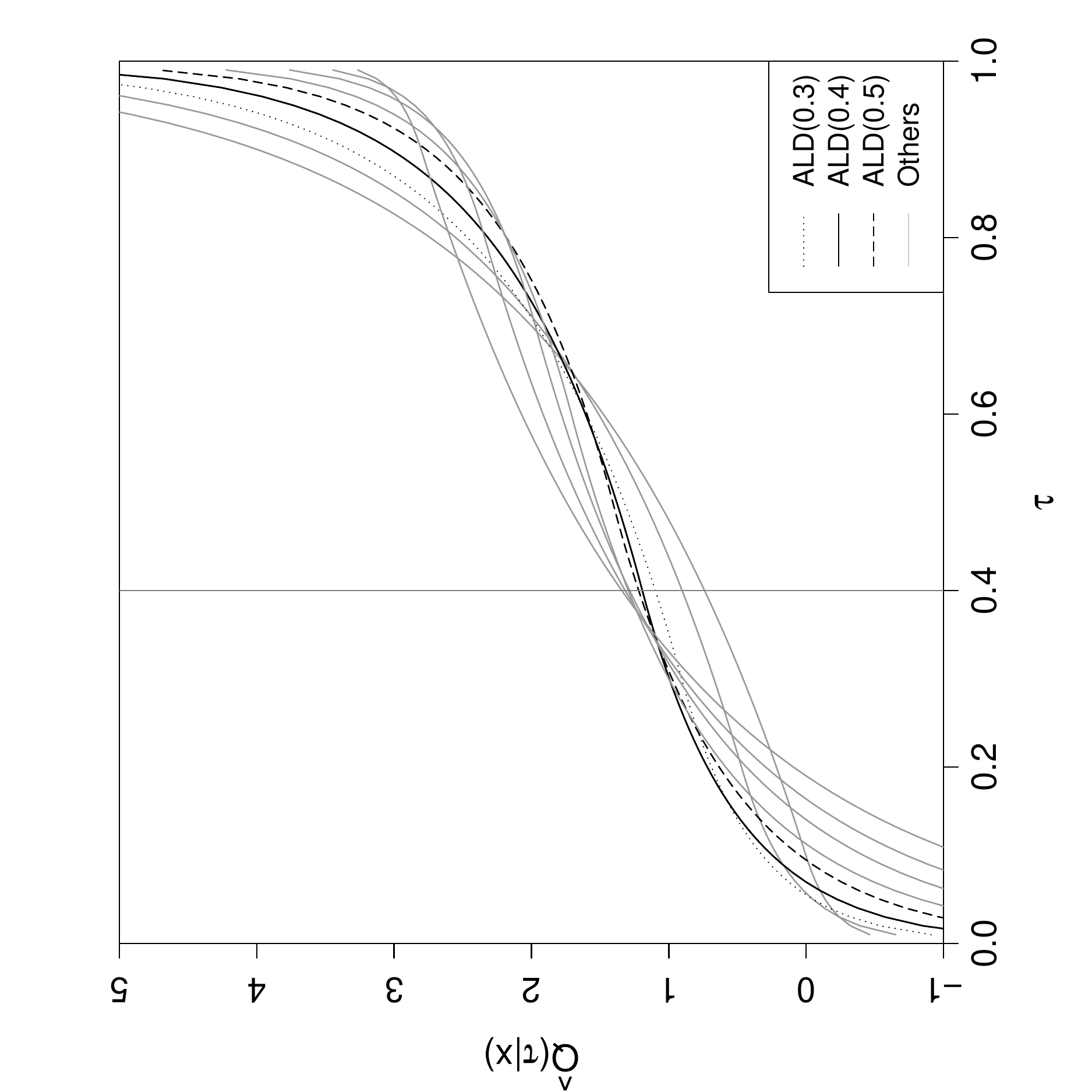}}
    \subfloat[\footnotesize{Second stage: Gaussian process regression} \label{idea_ald_b}]{\includegraphics[width=7.5cm,height=7.5cm,angle=-90]{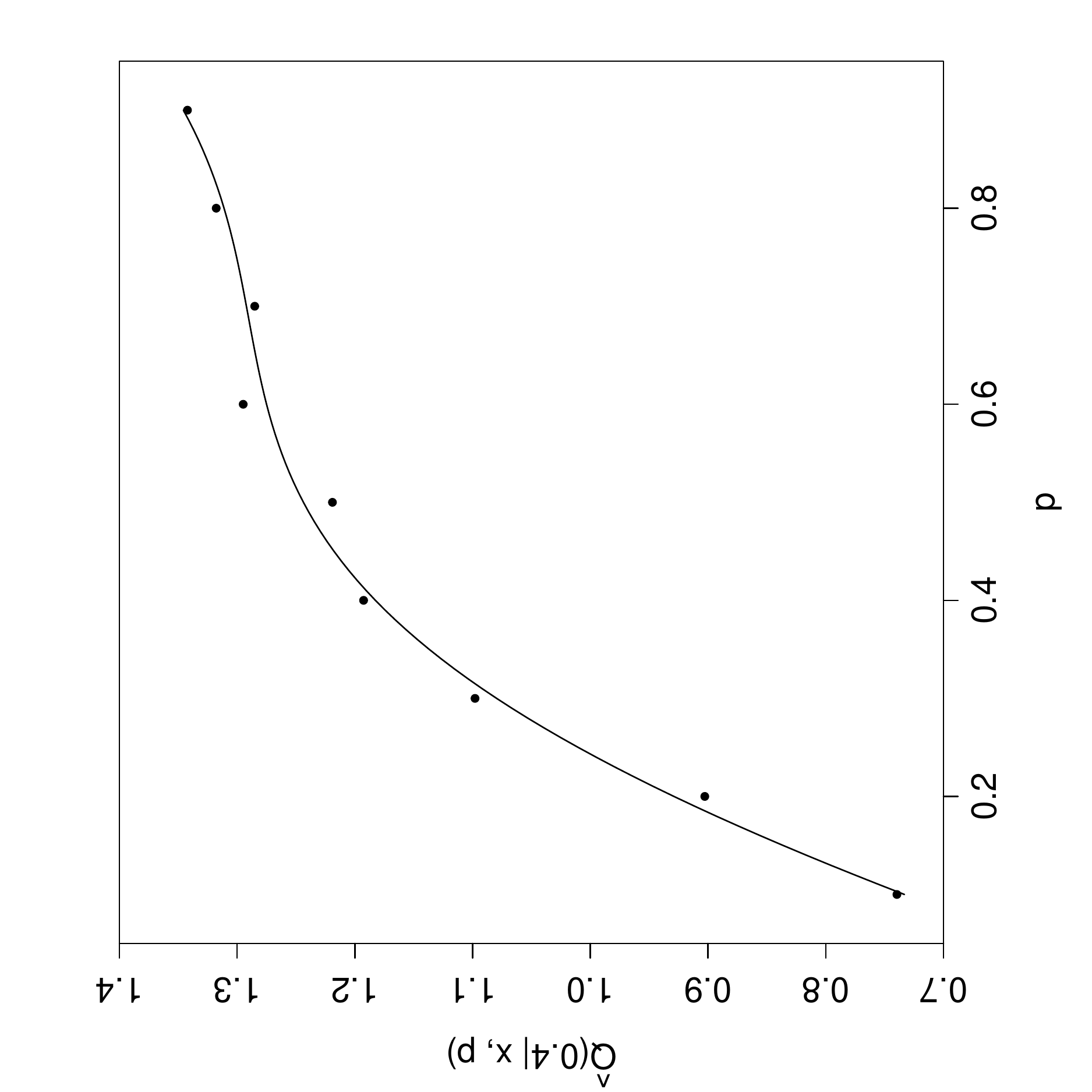}}
    \caption{Two-stage approach illustration:   $0.4$-th quantile estimation for a given value of $\mathbf{X}$.}
    \label{idea_ald}
\end{figure}

The performance of the regression adjusted quantile estimator will be studied on simulated and real data examples in the next sections.

\section{Simulation study}

\subsection{Linear quantile regression}
The proposed two-stage quantile model is compared to two noncrossing linear quantile regression methods, the constrained minimization approach of \cite{Bondell2010} and the semiparametric Bayesian model of \cite{Reich2013}. Codes for the first are available from the author's web page, whereas the second is implemented in BSquare package \citep{BSquare} in R \citep{Rmanual}. Our first stage estimation is carried out using the bayesQR package \citep{Rbayesqr} from R. Default uninformative priors were used throughout.
We considered four simulation designs studied by \cite{Reich2013}:

\begin{description}
 \item [Design 1.] $\beta_0(\tau)=\log[\tau/(1-\tau)]$, $\beta_1(\tau)=2$;
 \item [Design 2.] $\beta_0(\tau)=\text{sign}(0.5-\tau)\log{(1-2 \left|0.5-\tau \right|)}$, $\beta_1(\tau)=2\tau$;
 \item [Design 3.] $\beta_0(\tau)=\Phi^{-1}(\tau)$, $\beta_1(\tau)=2 \min{\{\tau-0.5,0\}}$;
 \item [Design 4.] $\beta_0(\tau)=2\Phi^{-1}(\tau)$, $\beta_1(\tau)=2 \min{\{\tau-0.5,0\}}$, $\beta_2(\tau)=2\tau$, $\beta_3(\tau)=2$, $\beta_4(\tau)=1$, $\beta_5(\tau)=0$;
\end{description}

For each design, we generate $U_i \overset{iid}{\sim} \text{Unif}(0,1), i=1,\dots,100$, and generate the $j$th covariate $X_{ij} \overset{iid}{\sim} \text{Unif}(-1,1)$. $Y_i = \beta_0(U_i)+ \sum_{j} X_{ij}\beta_j(U_i)$. Quantiles $\tau=0.05,0.06,\dots,0.95$ were fitted to the data. For the \cite{Reich2013} method, logistic base distribution was used with $4$ basis functions. For our first stage estimation, we fitted standard Bayesian quantile regression with $31500$ MCMC draws, thinning $m=30$ and burn-in of $1500$.

To compare the methods, $500$ data sets were simulated and the empirical root mean integrated squared error, $\text{RMISE}=\sqrt{1/n\sum_{i=1}^n{\{Q(\tau|X_i) - \widehat{Q}(\tau|X_i)\}^2}}$, was computed for each data set and $\tau$. Crossing occurred in $405$, $484$, $480$ and $500$ out of 500 data sets in the four designs. Figure \ref{sim_linear} presents the average RMISE for all designs and investigated methods. 

\begin{figure}[!ht]
    \centering
    \subfloat[\footnotesize{Design 1} \label{des1}]{\includegraphics[width=7.5cm,height=7.5cm,angle=-90]{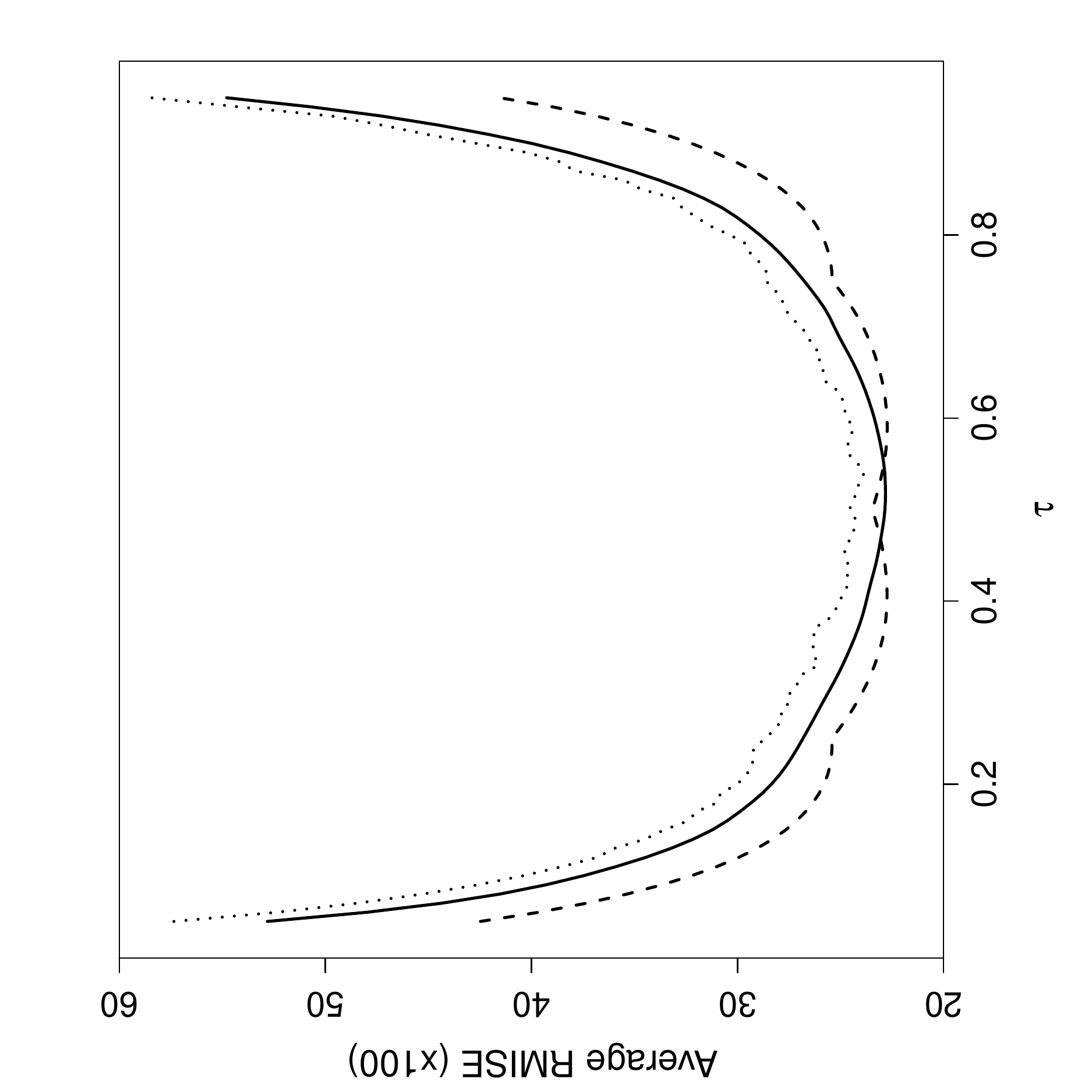}}
    \subfloat[\footnotesize{Design 2} \label{des2}]{\includegraphics[width=7.5cm,height=7.5cm,angle=-90]{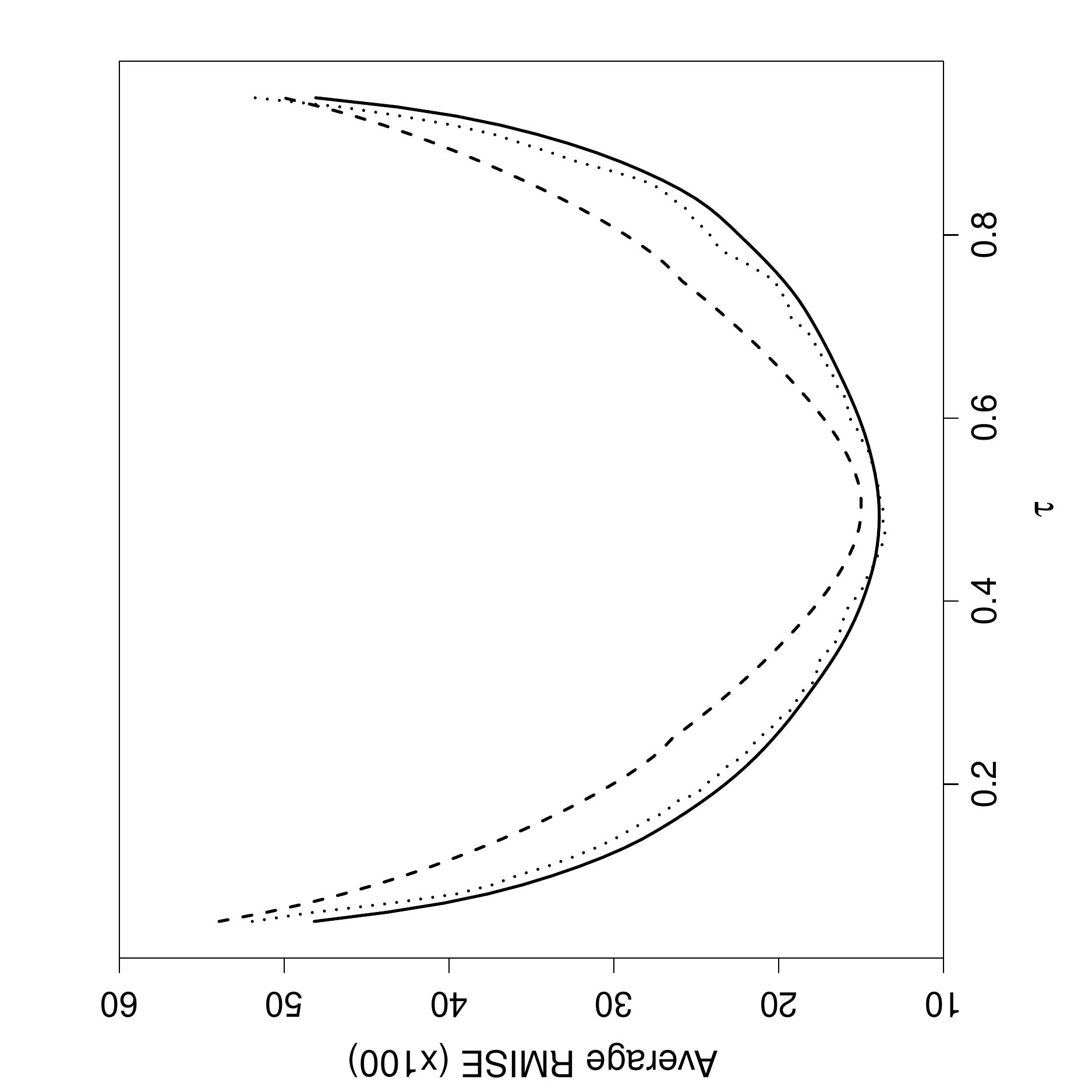}} \\
    \subfloat[\footnotesize{Design 3} \label{des3}]{\includegraphics[width=7.5cm,height=7.5cm,angle=-90]{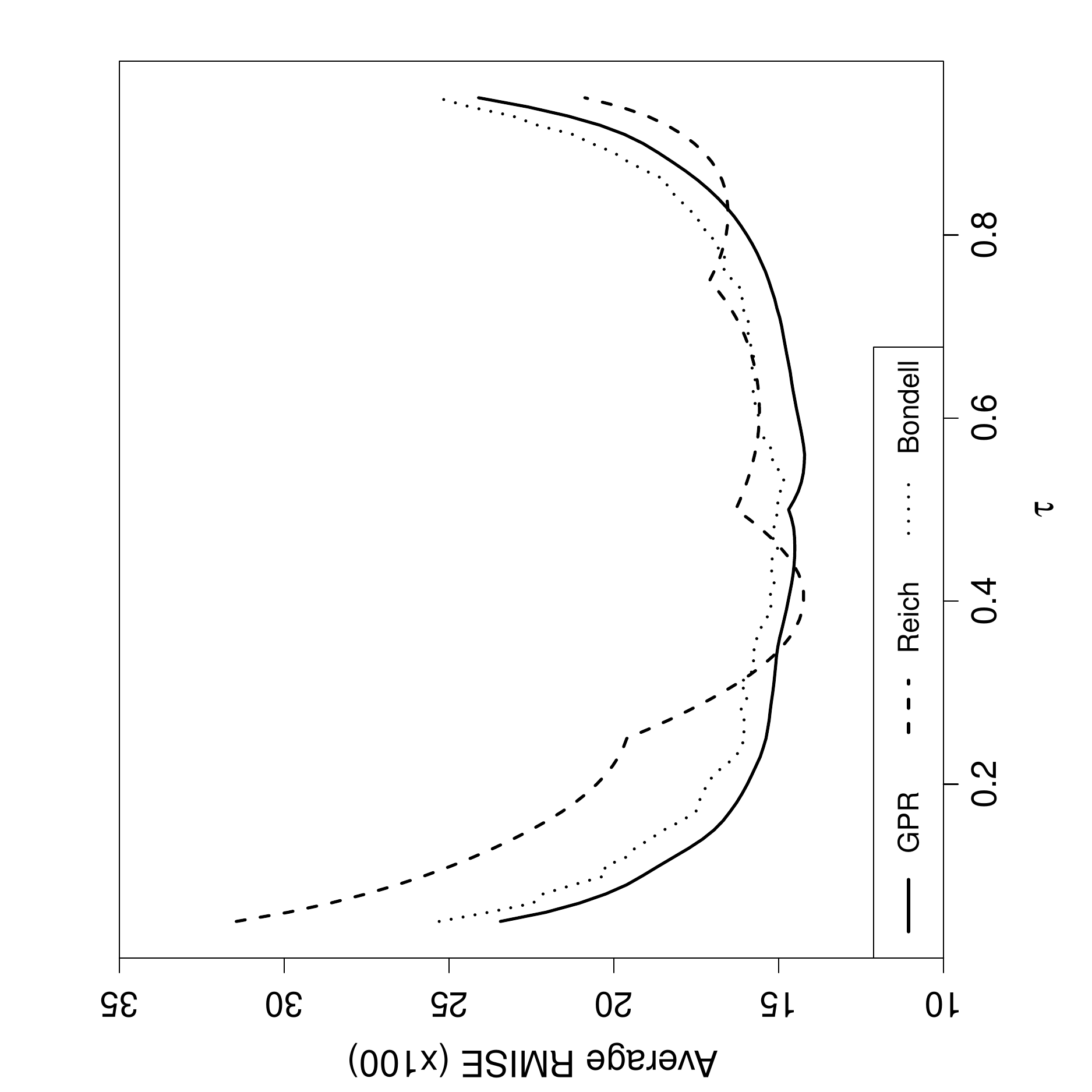}}
    \subfloat[\footnotesize{Design 4} \label{des4}]{\includegraphics[width=7.5cm,height=7.5cm,angle=-90]{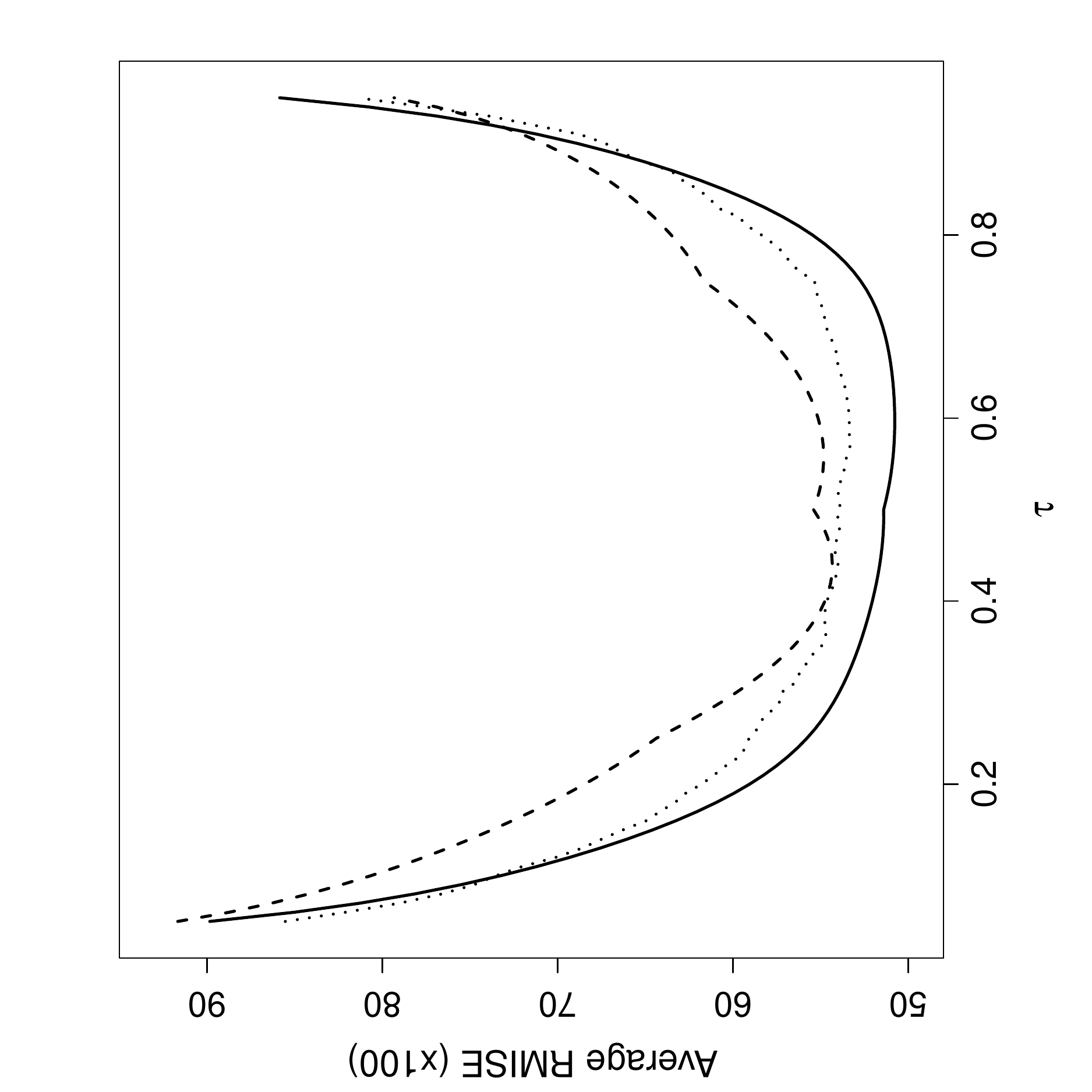}}
    \caption{Average RMISE ($\times 100$) over 500 data sets at $\tau=0.05,0.06,\dots,0.95$.}
    \label{sim_linear}
\end{figure}

The first design is very peculiar because $\beta_1$ does not vary with $\tau$. For this setting, the method of \cite{Reich2013} is significantly better than the others for the extreme quantiles, whereas at the interquartile range this difference is not so significant. Standard error plots of average RMISE were omitted, but values for all methods range from 0.6 (around the interquartile range) to 1.5 at the tails. This performance is expected as Design 1 satisfies the assumptions of the  \cite{Reich2013} model and, naturally, the correct model specification adds valuable information, especially where data availability is scarce. On the other hand, for Designs 2 and 3, where quantile regression is more appealing given that parameters are varying with $\tau$, the regression adjustment and \cite{Bondell2010} methods have similar performance, both outperforming  \cite{Reich2013} almost everywhere. For the multivariate scenario brought by Design 4, the proposed method has smallest RMISE than the others for most quantile levels and this difference is significant around the interquartile range.
Therefore, except for the simple scenario of Design 1, the proposed Gaussian process regression adjustment of standard Bayesian quantile estimates performs similar to or better than the simultaneous fitting methods of the other two approaches.

An initial drawback of the proposed approach is that the Gaussian process regression adjustment may jeopardize the linear relationship between the dependent and independent variables. Note that the final estimate is a weighted average of the induced quantiles \eqref{final_par}, whose weights depend on the covariance matrix $\Sigma'$ and, consequently, on the covariates $\mathbf{X}$. This compromises the linearity, particularly if the bandwidth is large. Nevertheless, using an approximation to a constant covariance matrix $\Sigma'$ with respect to $\mathbf{X}$ proves to perform well. 

For a given $\tau$ and $\mathbf{X}$, the elements of the diagonal covariance matrix $\Sigma'$ are the variances of the induced quantile posterior means $\widehat{Q}_s(\tau|\mathbf{X},p)$ for different auxiliary models ALD($p$). Figure \ref{cov_sigma} shows how these elements vary with $p$ for a sample from Design 3. The different grey curves are the covariances $\Sigma'$ for each observed $\mathbf{X}$. Notice that the grey curves share a similar shape across $p$ for different values of $\mathbf{X}$. Therefore, the covariance matrices $\Sigma'$ are approximately proportional to the mean covariance across $\mathbf{X}$ (black curve), denoted by $\bar{\Sigma}'$. Thus to maintain linearity, we may use $\bar{\Sigma}'$ to approximate the covariance matrix.

\begin{figure}[!ht]
    \centering
    \subfloat[\footnotesize{Quantile level $\tau=0.5$} \label{sigma50}]{\includegraphics[width=7.5cm,height=7.5cm,angle=-90]{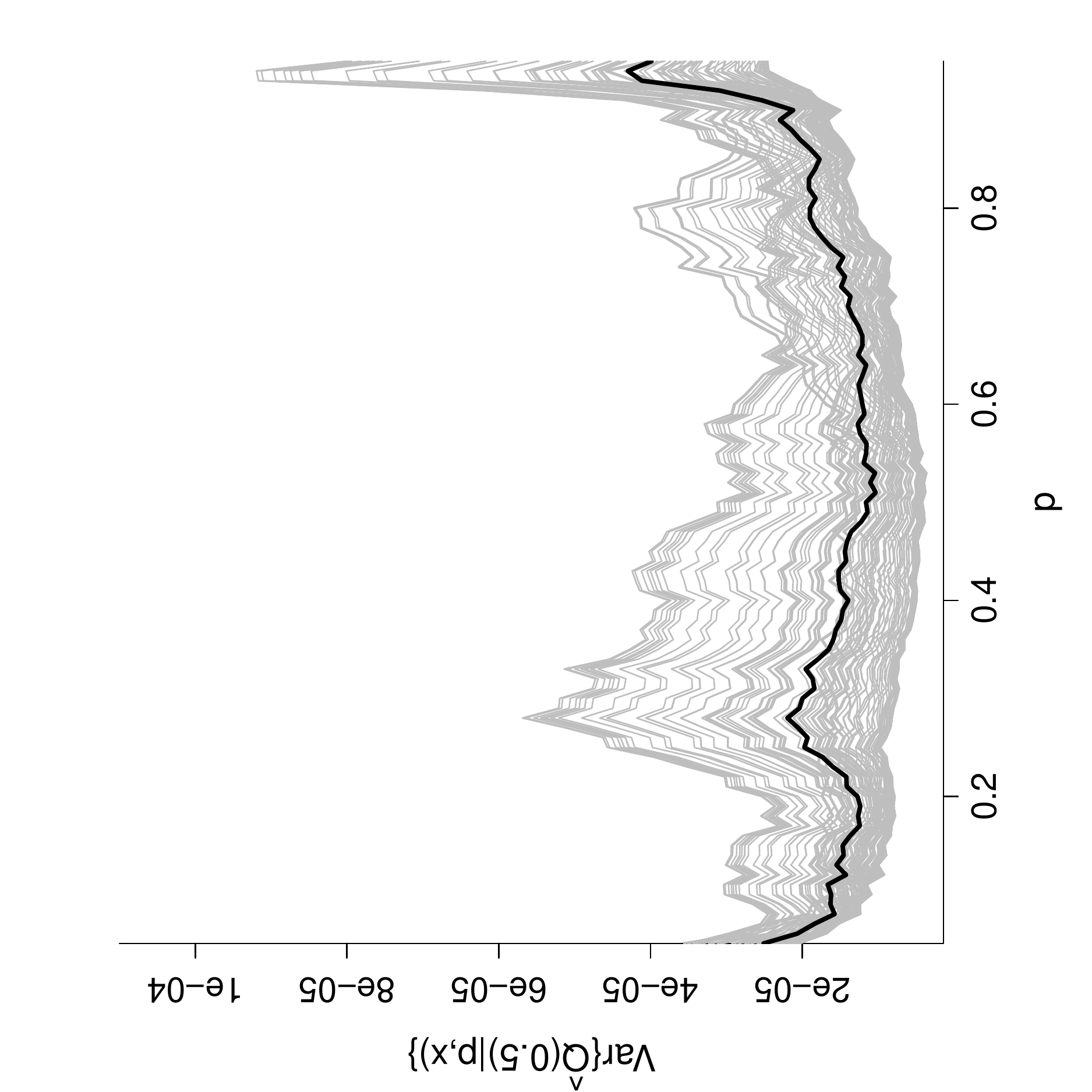}}
    \subfloat[\footnotesize{Quantile level $\tau=0.95$} \label{sigma95}]{\includegraphics[width=7.5cm,height=7.5cm,angle=-90]{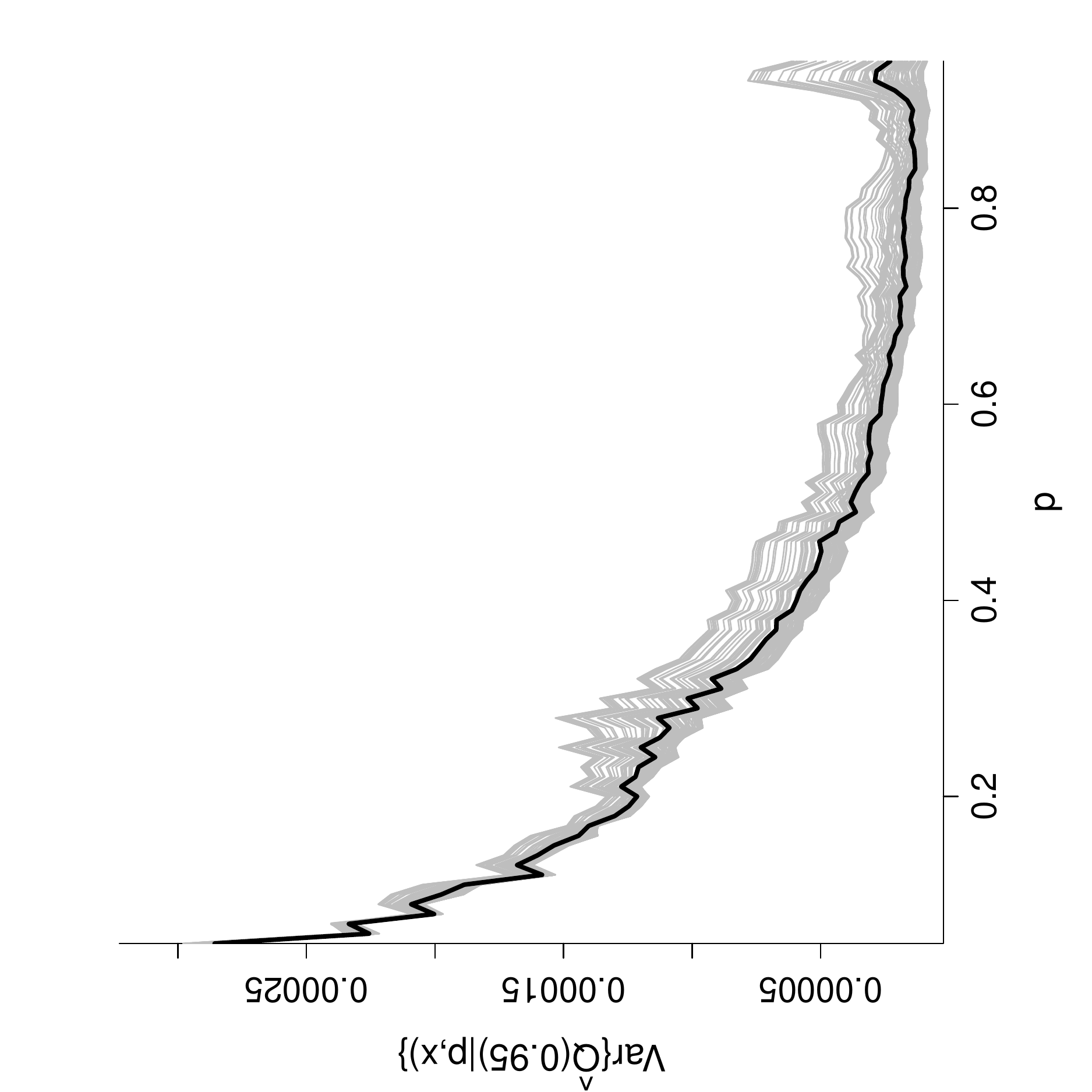}}
    \caption{Variance of induced quantile posterior means $\widehat{Q}_s(\tau|\mathbf{X},p)$ for all values of $\mathbf{X}$ (grey lines) and the mean variance across $\mathbf{X}$ (black line) for a sample from Design 3.}
    \label{cov_sigma}
\end{figure}

The average RMISE for the proposed linear Gaussian process regression adjustment (LGPR) coincides with the previous GPR results (Figure \ref{sim_linear}) for all simulated 
designs.  From Figure \ref{xy_betacte} we can see that the final estimates from LGPR are very similar to the ones obtained through GPR even when the bandwidth is very large. Therefore the approximation used for linear regression appears to be satisfactory, producing very similar results from the original method while holding the linearity needed for parameter interpretation purposes.

\begin{figure}[!ht]
    \centering
    \subfloat[\footnotesize{Median bandwidth ($b=0.1$)} \label{xy1}]{\includegraphics[width=7.5cm,height=7.5cm,angle=-90]{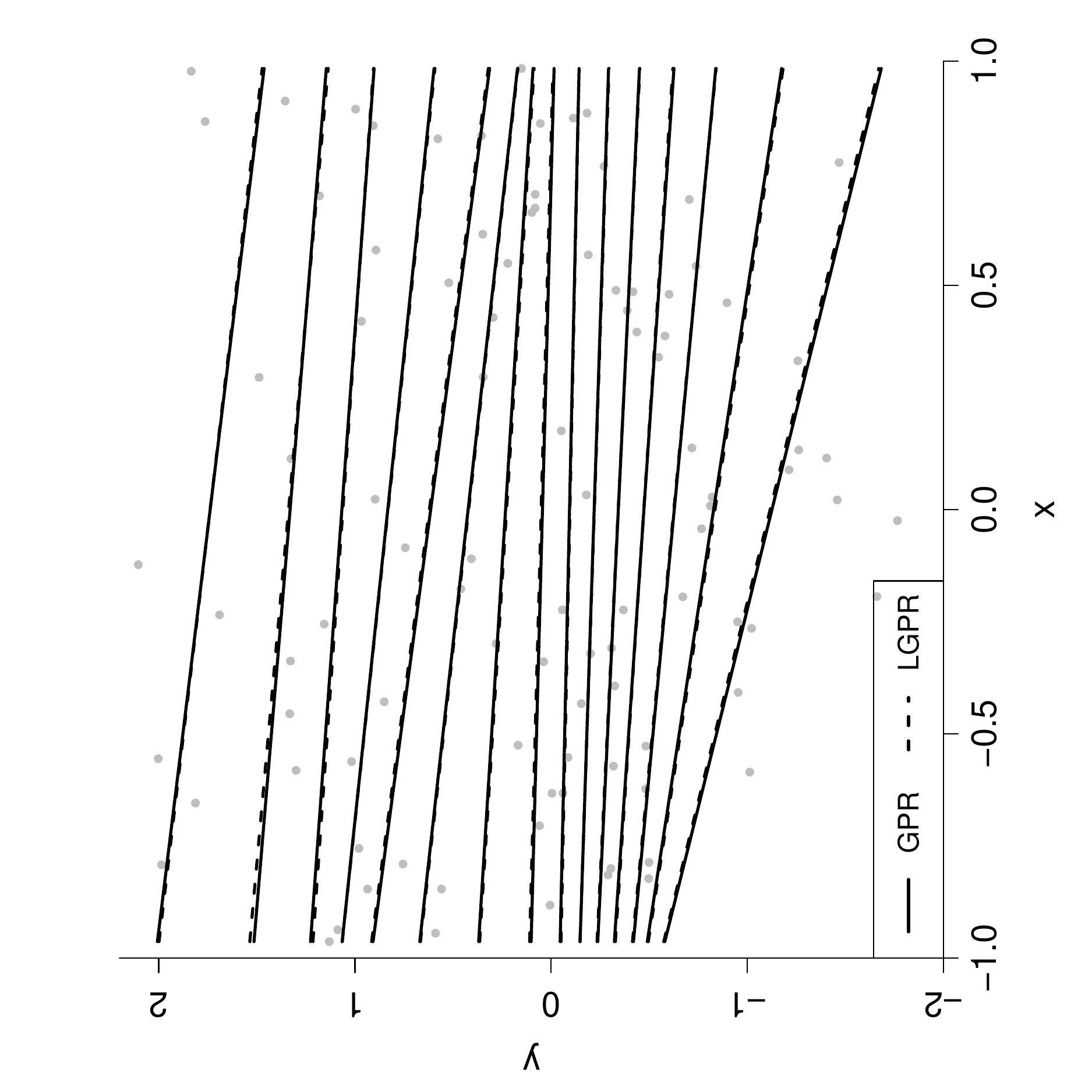}}
    \subfloat[\footnotesize{Largest bandwidth ($b=2000$)} \label{xy2}]{\includegraphics[width=7.5cm,height=7.5cm,angle=-90]{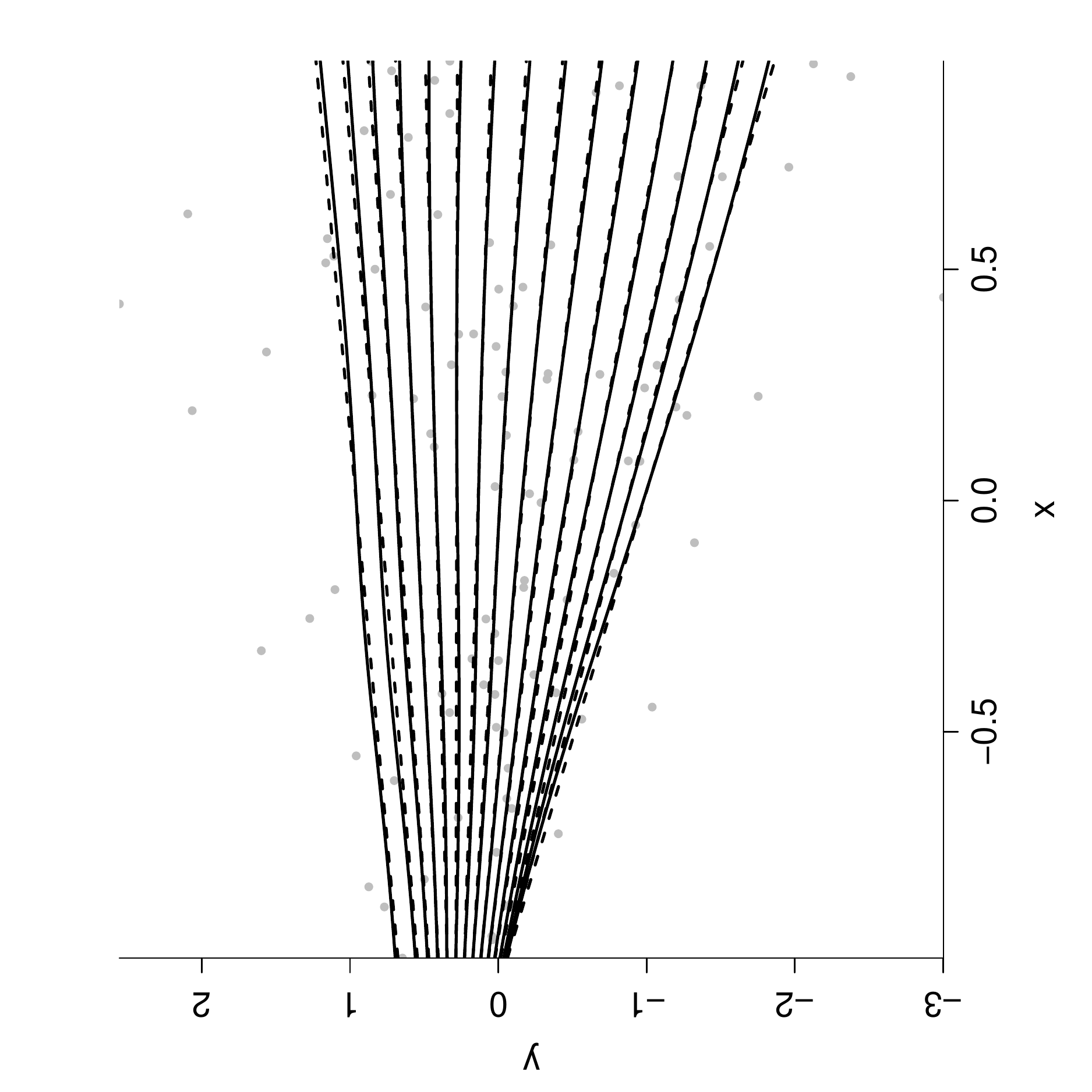}}
    \caption{Quantile regression estimates using Gaussian process regression adjustment (GPR) and linear Gaussian process regression adjustment (LGPR) for two different samples from Design 3 at quantile levels $\tau=0.10,0.16,\dots,0.94$.}
    \label{xy_betacte}
\end{figure}

\subsection{Nonparametric quantile regression}
For nonparametric quantile regression, the frequentist approach of \cite{Bondell2010} for fitting linear splines will be considered. 
However, as their codes are not available for the nonparametric case, we will analyse the same simulation designs proposed by them to have their results for comparison. Let $y_i = f(x_i) + g(x_i)\epsilon_i$ be a heteroscedastic error model, then Designs 5 and 6 are given by the following choices of mean and covariance functions:

\begin{description}
 \item [Design 5.] $f(x)=0.5+2x+\text{sin}(2\pi x-0.5)$, $g(x)=1$;
 \item [Design 6.] $f(x)=3x$, $g(x)=0.5+2x+\text{sin}(2\pi x-0.5)$;
\end{description}

Similarly, for each design, $i=1,\dots,100$ samples were generated, given that $X_{i} \overset{iid}{\sim} \text{Unif}(0,1)$ and $\epsilon_i \overset{iid}{\sim} \text{N}(0,1)$. Quantile levels $\tau=0.05,0.06,\dots,0.95$ were estimated. For the first stage regression adjustment, cubic splines with 25 equally spaced knots were fitted to the data using standard Bayesian quantile regression from bayesQR package (same MCMC configuration as previous simulation designs). Again, we use uninformative priors. Linear B-splines with knots at each data point were used by \cite{Bondell2010}. We simulated $500$ data sets and all samples presented crossing issues when fitted with standard Bayesian quantile regression. Results are presented in Figure \ref{sim_splines}.

\begin{figure}[!ht]
    \centering
    \subfloat[\footnotesize{Design 5} \label{des5}]{\includegraphics[width=7.5cm,height=7.5cm,angle=-90]{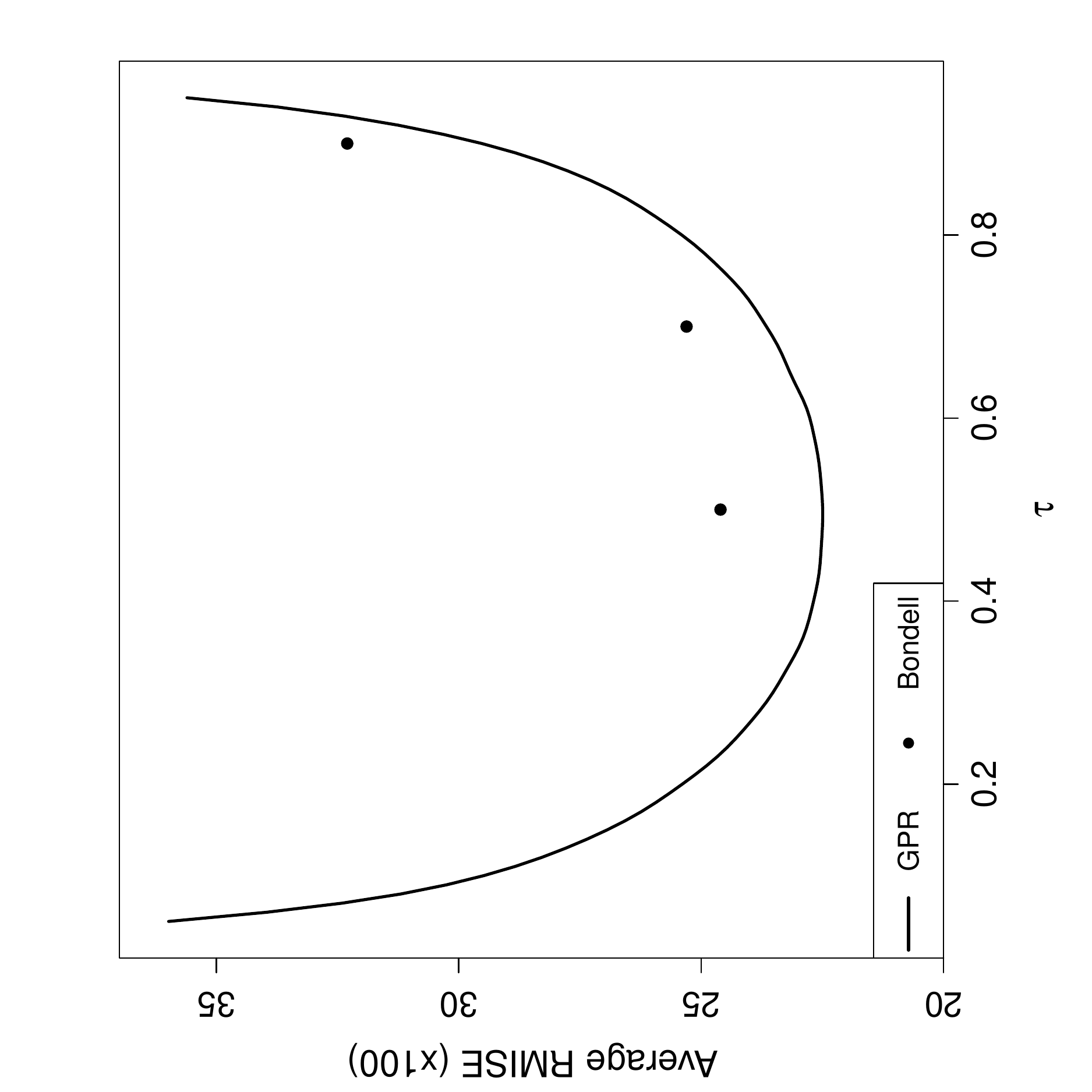}}
    \subfloat[\footnotesize{Design 6} \label{des6}]{\includegraphics[width=7.5cm,height=7.5cm,angle=-90]{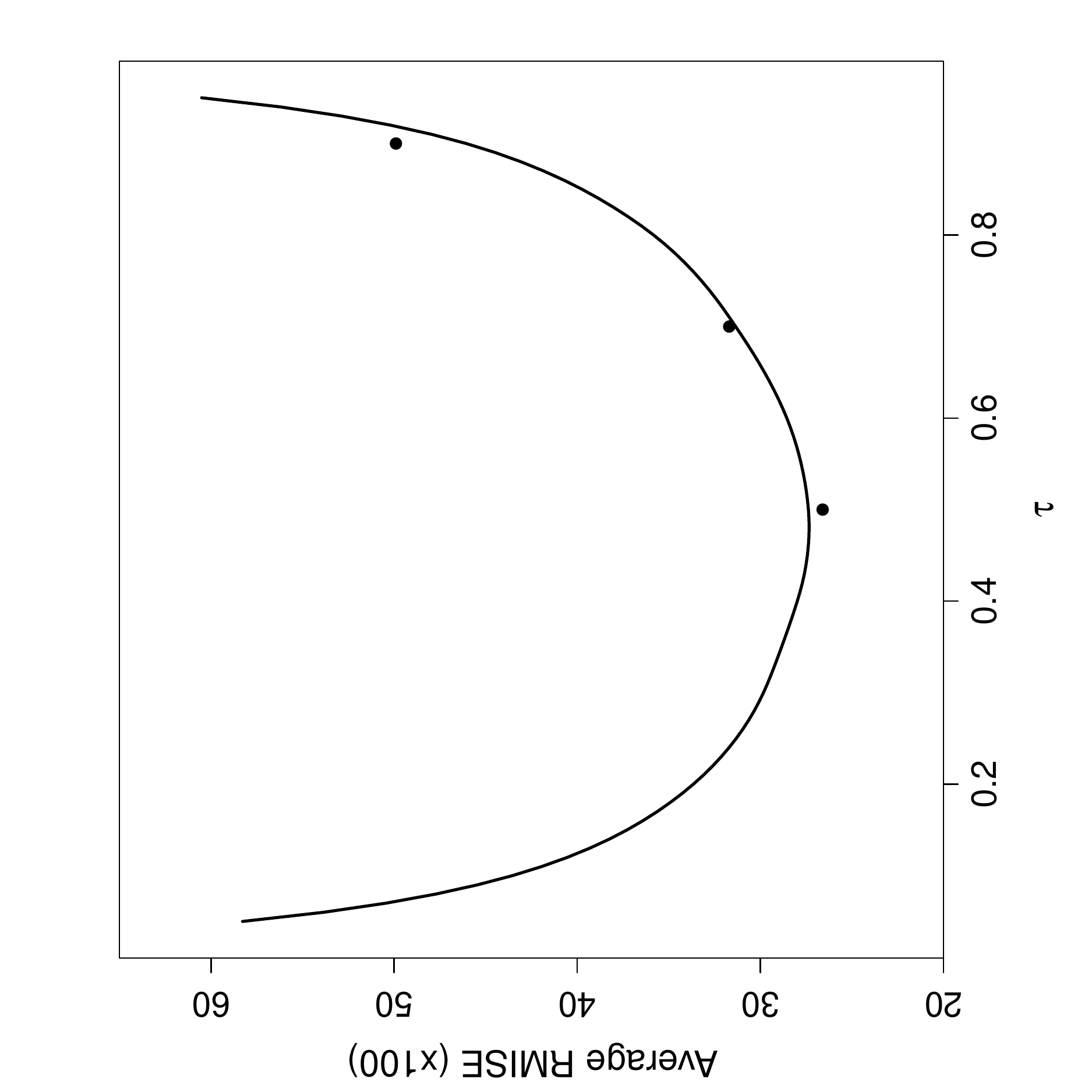}}
    \caption{Average RMISE ($\times 100$) over 500 data sets at  $\tau=0.05,0.06,\dots,0.95$. Results for Bondell's method for $\tau=0.5,0.7,0.9$ are from \cite{Bondell2010}.}
    \label{sim_splines}
\end{figure}

For Design 5 and $\tau=0.5,0.7,0.9$, the magnitude of the standard errors of average RMISE of both methods are below 0.5 (standard error plots omitted), suggesting that regression adjustment has significantly smaller RMISE than the constrained minimization approach. Whereas for the complex covariance function brought by Design 6, both methods have similar performance. Furthermore, as the regression curves are, in general, expected to be smooth, cubic splines interpolation is preferable than linear splines, see Figure \ref{sim_splines_xy}. However, cubic splines is not supported by \cite{Bondell2010}.

\begin{figure}[!ht]
    \centering
    \subfloat[\footnotesize{Design 5} \label{des5.xy}]{\includegraphics[width=7.5cm,height=7.5cm,angle=-90]{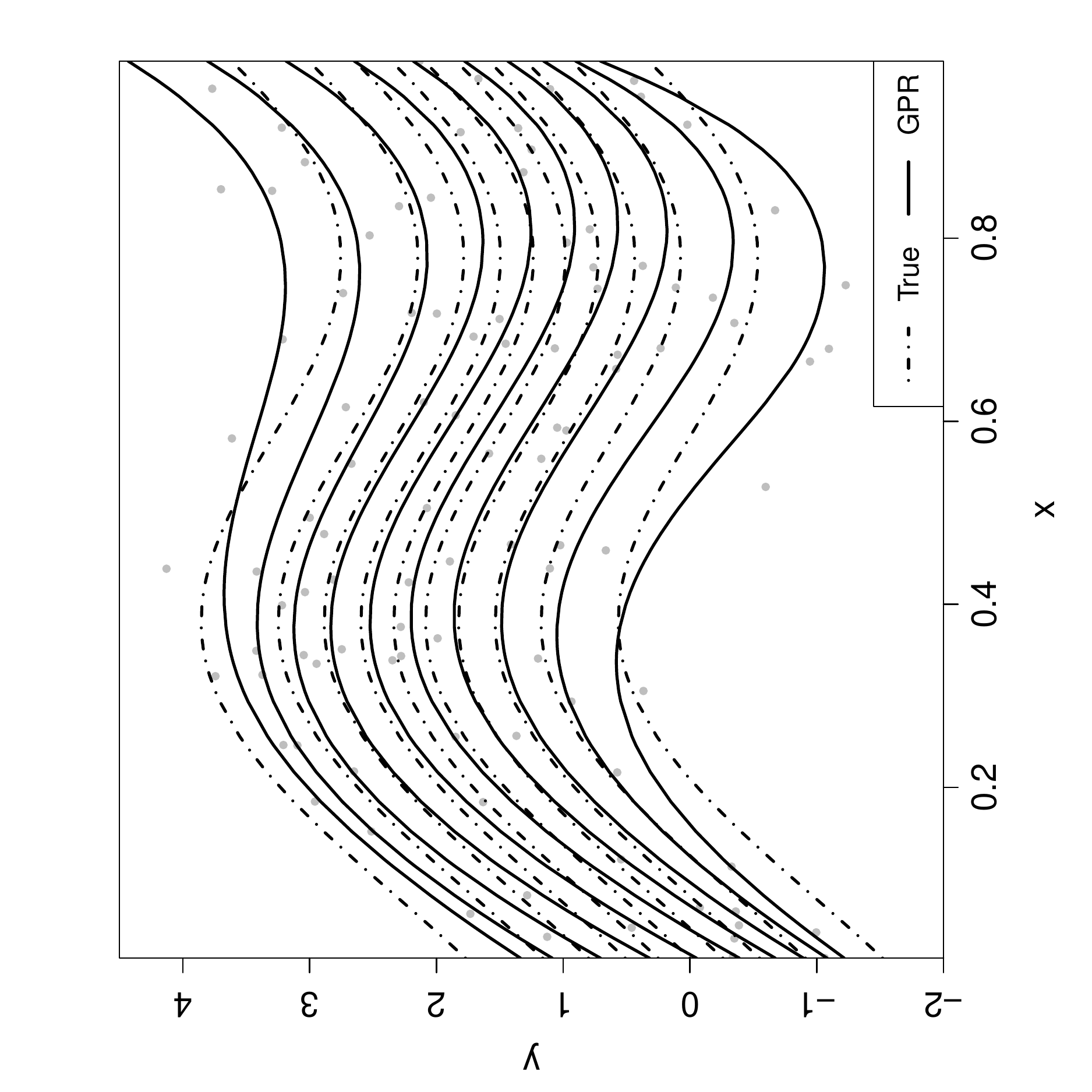}}
    \subfloat[\footnotesize{Design 6} \label{des6.xy}]{\includegraphics[width=7.5cm,height=7.5cm,angle=-90]{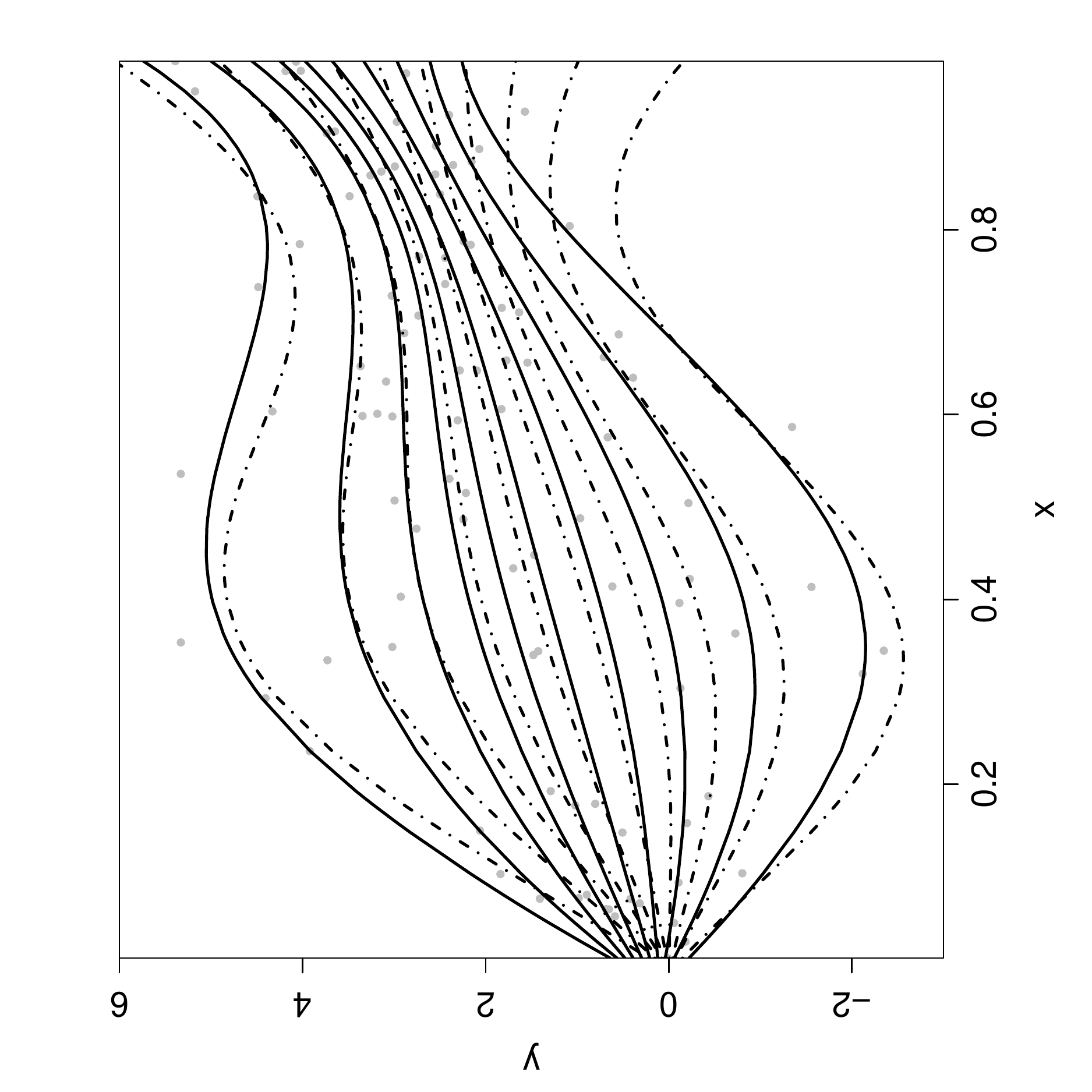}}
    \caption{True and estimated conditional quantile functions for quantile levels $\tau=0.05,0.15,\dots,0.95$.}
    \label{sim_splines_xy}
\end{figure}

Therefore, the two stage approach allows for more flexibility to handle any nonparametric quantile function in the first stage, and results are better than, or at least similar to, competitive approaches.

\section{Real data applications}
Quantile regression is widely used in medicine and environmental sciences. Both applications will be explored in this section through two real data examples.

\subsection{Immunoglobulin-G data set}
Centile charts are adopted in medicine to establish reference ranges in order to identify unusual subjects. However, as interest lies in estimating many quantiles of the response distribution, the estimates often cross.
In the search for reference ranges to help diagnose immunodeficiency in infants, \cite{Isaac1983} measured the serum concentration of immunoglobulin-G (IgG) in 298 preschool children. This famous data set will be used here to estimate IgG conditional quantile levels $\tau=0.005, 0.01, ..., 0.995$. A quadratic model in age is used to fit the data due to the expected smooth change of IgG with age. See \cite{Isaac1983}, \cite{Yu2001} and \cite{Kottas2009}.

Figure \ref{igg} presents the results before and after the Gaussian process regression adjustment. For ease of visualization, fewer quantile curves are drawn from $\tau=0.05$ to $\tau=0.95$ (step size equals 0.05), whereas at the tails, where most of the crossing occurs, all estimated quantiles are plotted. The highlighted curves have quantile levels identified on its right side for reference purposes.

\begin{figure}[!ht]
    \centering
    \subfloat[\footnotesize{Standard Bayesian quantile regression} \label{igg2}]{\includegraphics[width=7.5cm,height=7.5cm,angle=-90]{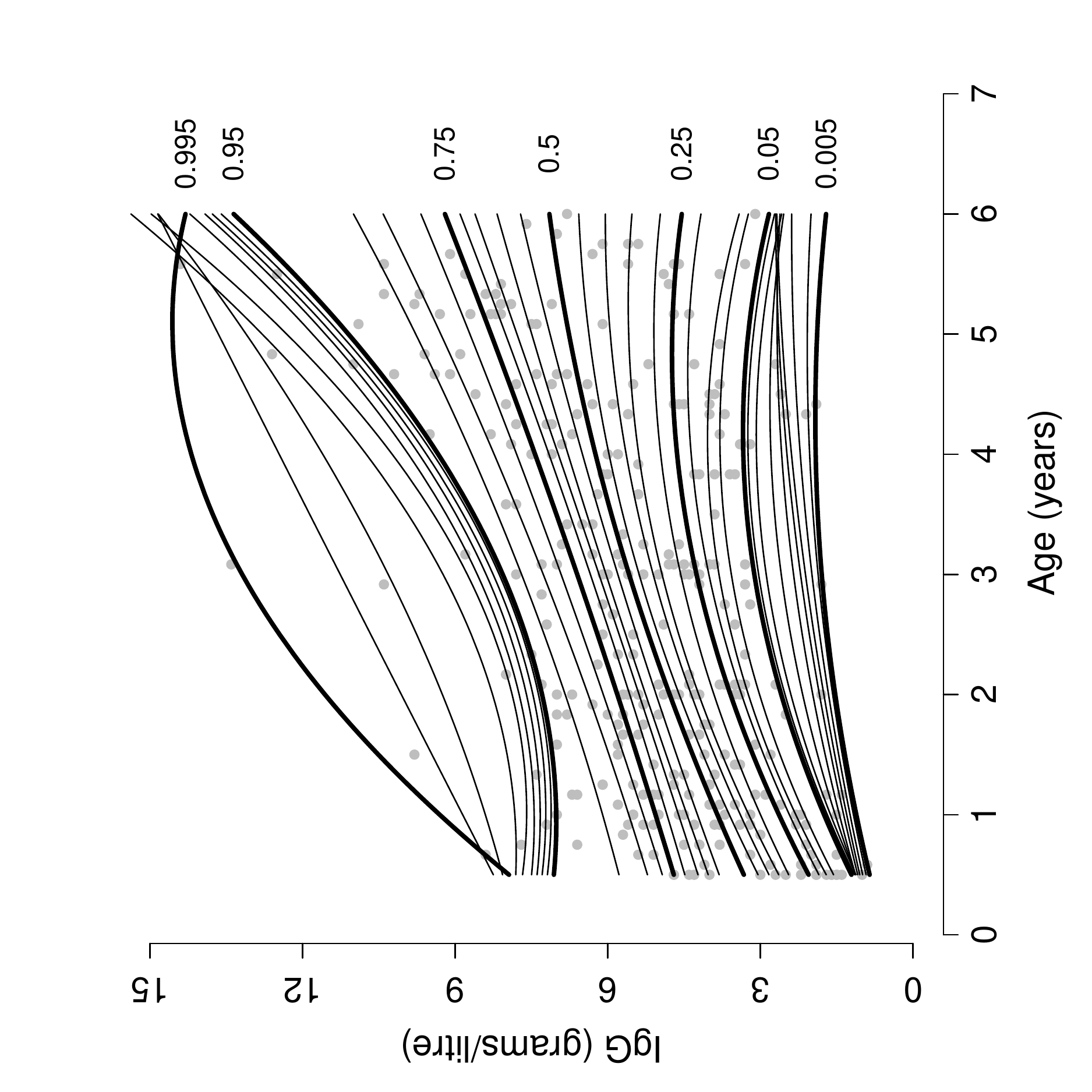}}
    \subfloat[\footnotesize{Gaussian process regression adjustment} \label{igg3}]{\includegraphics[width=7.5cm,height=7.5cm,angle=-90]{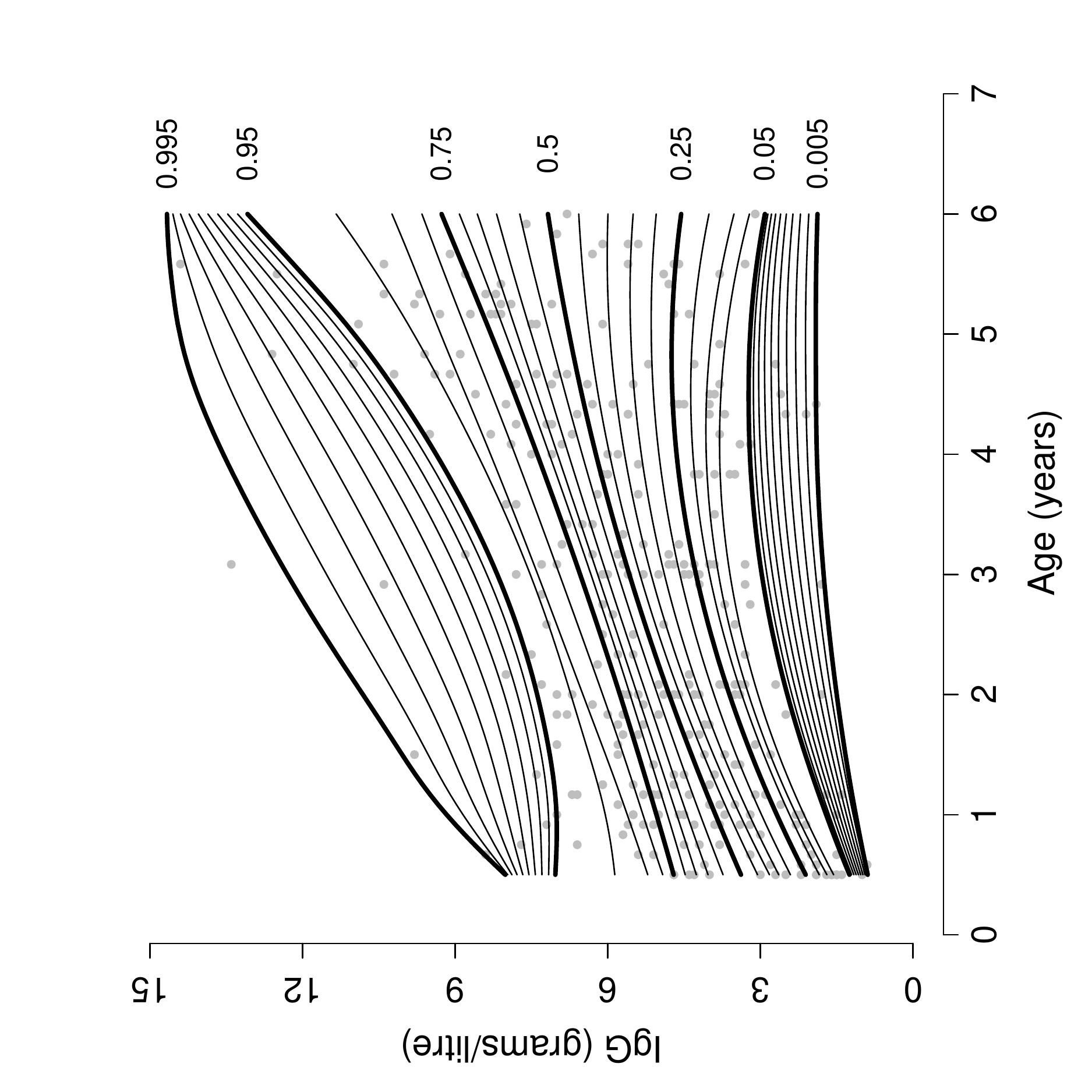}}
    \caption{Growth chart of serum concentration of immunoglobulin-G for young children.}
    \label{igg}
\end{figure}

The independent fitting of standard Bayesian quantile regression (Figure \ref{igg2}) displays many crossing curves. Furthermore, the comparison of quantile curve estimate for $\tau=0.995$ and its neighbour $\tau=0.99$ demonstrates the consequences of not borrowing strength. Although both quantile levels are remarkably close, we get very different estimates when fitting them separately, which does not look very realistic and leads to the extreme case of crossing. The Gaussian process regression adjustment corrects the crossing by borrowing information from nearby quantiles (Figure \ref{igg3}). Therefore, the final estimates not only respect the monotonicity constraint, but are smoother than the initial ones. The estimated conditional quantile function for children with 6 years old, previously presented in Figure \ref{igg1}, evidences the monotonicity and smoothing effect of the two-stage approach. Hence, better quantile estimates are provided here, without compromising the simplicity or flexibility of the standard Bayesian approach.

\subsection{Global mean sea level variation data set}
From satellite radar altimeters measurements, \cite{Nerem2010} constructed a climate data record of global mean sea level change ($\Delta$MSL), which is defined as ``the area-weighted mean of all of the sea surface height anomalies measured by the altimeter in a single, 10-day satellite track repeat cycle''. The data set consists of $762$ observations dating from 1992 to 2014 (for data processing details see \cite{Nerem2010}).

Sea level rise is closely related to climate change and its study is paramount, specially for habitants of coastal and island regions, where extreme events are likely to become more frequent. Quantile regression can provide a greater picture of the global mean sea level distribution whilst modelling several quantile levels. For this study, we used the regression adjustment method to estimate 19 regression curves, $\tau=0.05, 0.10, ..., 0.95$. In the first stage (Figure \ref{seai}), nonparametric Bayesian quantile regression was fitted using cubic splines. Only the first decade is presented for a better visualisation of the crossing issue. We used the model described by \citep{Yanan2012}, where regression splines of any degree can be fitted, with unknown knot number and location, based on the ALD distribution, codes were provided by the authors. Alternative methods for spline fitting can also be found in \cite{chenyu09} and \cite{thompsoncmrs10}. 

\begin{figure}[!ht]
    \centering
    \subfloat[\footnotesize{Standard Bayesian quantile regression} \label{seai}]{\includegraphics[width=7.5cm,height=7.5cm,angle=-90]{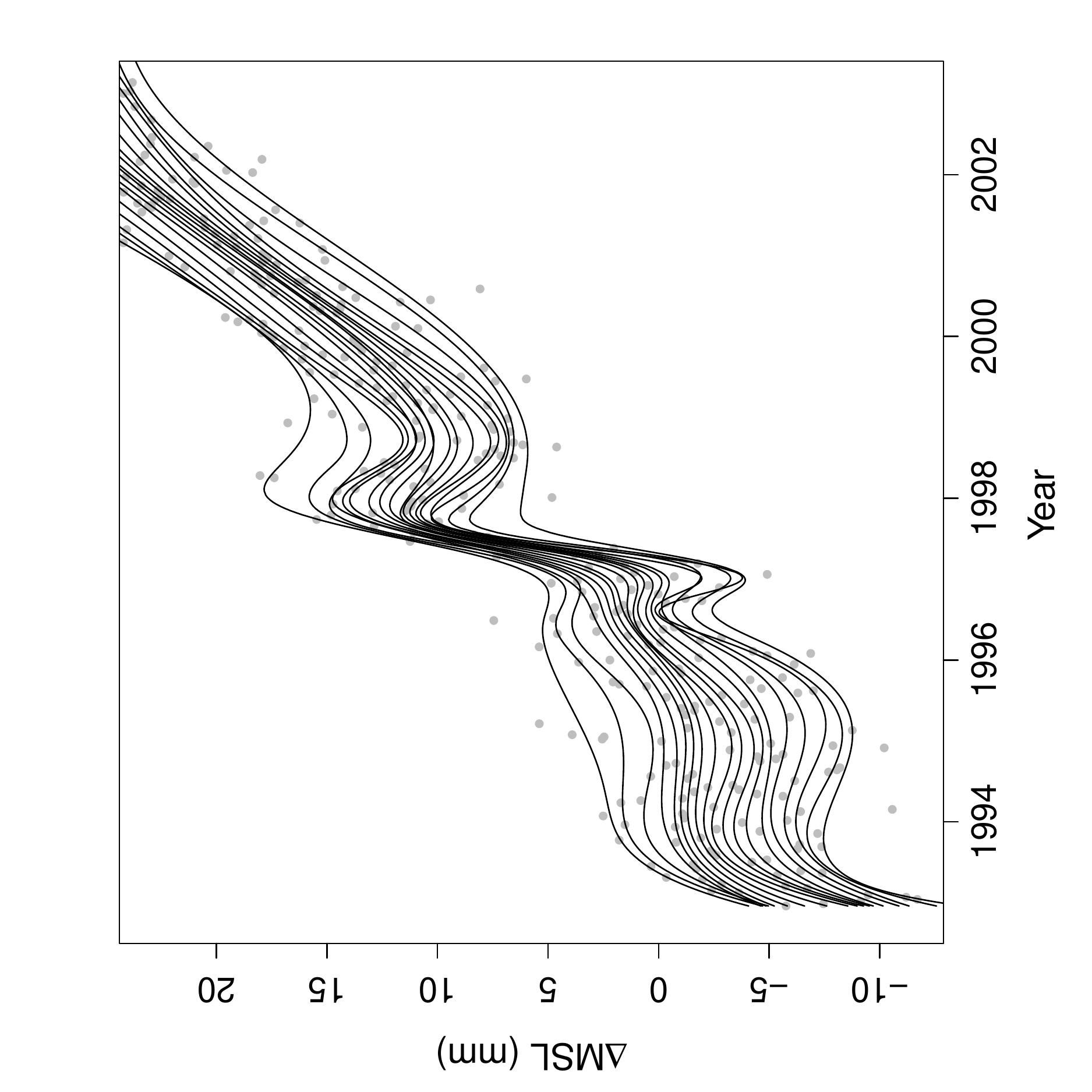}}
    \subfloat[\footnotesize{Gaussian process regression adjustment} \label{seaf}]{\includegraphics[width=7.5cm,height=7.5cm,angle=-90]{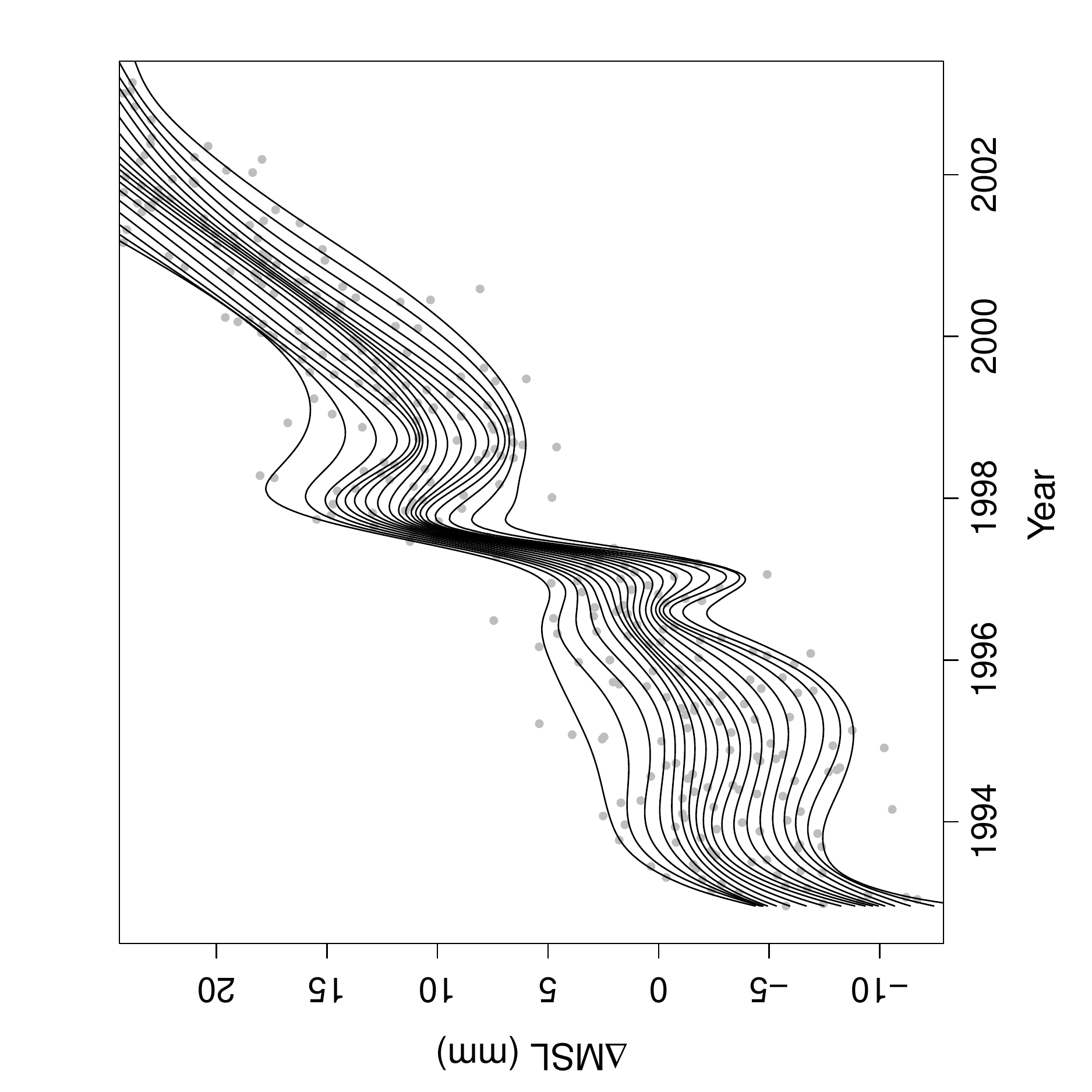}}
    \caption{Quantile regression of global mean sea level variation for $\tau=0.05, 0.10, ..., 0.95$.}
    \label{sea}
\end{figure}

From Figure \ref{seai}, we notice that, as a consequence of the nonparametric fit, crossing becomes much more likely to occur. After adding the monotonicity constraint, the regression adjustment provides improved estimates by getting rid of the crossing and adding some nice smoothness (Figure \ref{seaf}). This study corroborates that sea level is overall increasing since 1992, although the rate of the change is not constant over the period (plots for the second decade were ommited, but increasing pattern persists and $\Delta$MSL reaches 70 millimeters in 2014). As seasonal variations of the sea level were subtracted from this data set, the oscillating inter-annual periods are generally associated with El Nino and La Nina events \citep{Nerem2010}. Quantile modelling adds important information to the monitoring process of global sea levels and ultimately helps to mitigate the impact of extreme events.

\section{Concluding remarks}
This paper presents a two-stage approach to noncrossing Bayesian quantile regression. Following \cite{Yu2001}, the asymmetric Laplace distribution is used in the first stage as an auxiliary likelihood to fit several quantiles separately. A better exploration of this set of auxiliary fittings is undertaken in the second stage, in which the initial estimates are adjusted by borrowing strength from nearby quantiles using Gaussian process regression. The procedure is computationally simple as initial quantiles are fitted independently and no MCMC is needed in the second stage. Monotonicity constraints are also handled easily through the estimation of a single parameter. In addition, final estimates reduces to the initial ones if noncrossing exists and posterior consistency is maintained in case it holds initially. 

Simulation studies demonstrate the flexibility of the two-stage approach in modelling complex quantile functions, providing in general smaller RMISE than alternative methods. Moreover, it handles both linear and non-linear quantile regression curves. Although an approximation to the covariance matrix is needed in the linear case, its performance is very satisfactory. The possibility to choose any nonparametric curve to model the data again shows the great versatility of the proposed method. Indeed the algorithm can be modified to accommodate models other than the ALD in the first stage. Additional smoothness of the quantile estimates is another attractive property as illustrated in real data analyses. 

Being an adjustment of initial estimates, the performance of the proposed and standard approaches are naturally related. Flexibility and low RMISE, for instance, are great features directly incorporated. However, caution is needed in the case where the initial estimates perform too badly, which can occur when sample size is small and, consequently, extreme quantiles are fitted very poorly. The two-stage approach will still correct the crossing issues, but at a cost of an undesirable high bandwidth. In these circumstances, simultaneous quantile estimation should be preferred.

\subsection*{Acknowledgements}
TR is funded by the CAPES Foundation via the Science Without Borders program (BEX 0979/13-9).

\section{Appendix}

\begin{proof}[Proof of Proposition 2]
Let $q_{pt}$ be a short notation for observed induced quantile sample $Q^{(t)}(\tau|\mathbf{x},p_p)$ and $\bar{q}_p$ be the observed induced quantile posterior mean $\widehat{Q}_s(\tau|\mathbf{X},p_p)$. Then, the likelihood function for the model with all MCMC samples \eqref{gp_model} can be factorized as follows
\begin{equation*}
\begin{aligned}
p(\mathbf{q}_{pt}|\mathbf{g}(p_p), \boldsymbol{\sigma}^2_p) &= \prod_{p=1}^{P} \prod_{t=1}^{T}\frac{1}{\sqrt{2\pi\sigma_p}} \exp{ \left\{ - \frac{1}{2\sigma^2_p}\left( q_{pt} - g(p_p) \right) ^2 \right\} } \\
&= \prod_{p=1}^{P} (2\pi\sigma_p)^{-\frac{T}{2}} \exp{ \left\{ - \frac{1}{2\sigma^2_p} \sum_{t=1}^{T} {\left( q_{pt} - g(p_p) \right) ^2 } \right\} } \\
&= \prod_{p=1}^{P} (2\pi\sigma_p)^{-\frac{T}{2}} \exp{ \left\{ - \frac{1}{2\sigma^2_p} \sum_{t=1}^{T} {\left( q_{pt} - \bar{q}_p \right) ^2 } \right\} } \prod_{p=1}^{P} \exp{ \left\{ - \frac{T}{2\sigma^2_p} \left( g(p_p) - \bar{q}_p \right) ^2 \right\} }.
\end{aligned}
\end{equation*}
Therefore Fisher-Neyman factorization theorem implies that induced quantile posterior means $T(\mathbf{q}_{pt})=(\bar{q}_1, \dots, \bar{q}_P$) are jointly sufficient for $(g(p_1), \dots, g(p_P))$. Consequently, the likelihood for the model with all MCMC observations is proportional to the likelihood for the model with induced quantile posterior means, i.e. $p(\mathbf{q}_{pt}|\mathbf{g}(p_p), \boldsymbol{\sigma}^2_p) \propto p(\bar{\mathbf{q}}_p|\mathbf{g}(p_p), \boldsymbol{\sigma}^2_p)$. Furthermore, as the priors for $g(p)$ are the same for both models, then they have the same posterior and predictive posterior distributions for $g(p)$. 

Therefore, considering that the predictive posterior distribution for $g'(p{=}\tau)$ is $\mathcal{N} (\mu', {\sigma'}^2)$, so is the predictive distribution for $g(p{=}\tau)$. Lastly, as $Q_*(\tau|x,p{=}\tau,q_{pt}) = g_*(p{=}\tau) + \epsilon$, the predictive posterior distribution for $Q(\tau|x,p)$ is $Q_*(\tau|x,p{=}\tau,q_{pt}) \sim \mathcal{N} (\mu', {\sigma'}^2+ \sigma^2(\tau|\mathbf{X},p{=}\tau))$.
\end{proof}

\bibliography{bibli}

\end{document}